\documentclass[letterpaper,11pt]{article}
\usepackage[cmex10]{amsmath}
\usepackage{amsfonts}
\usepackage{amssymb}
\usepackage{ifthen}
\usepackage{amsthm}
\usepackage{fullpage}
\usepackage{multirow}
\usepackage[table]{xcolor}
\usepackage{colortbl}

\usepackage{appendix}

\usepackage{algorithmicx}
\usepackage{algorithm}
\usepackage[noend]{algpseudocode}
\usepackage{array}

\usepackage{xspace}
\usepackage[bookmarks=false]{hyperref}
\usepackage[all]{hypcap}

\usepackage{url}

\usepackage{color, soul}

\usepackage[absolute,overlay]{textpos}
\usepackage{tikz}
\tikzstyle{block}=[draw opacity=0.7,line width=1.4cm]
\usetikzlibrary{arrows,decorations,positioning,backgrounds,fit,shapes,matrix}

\usetikzlibrary{backgrounds}
\usepackage{color}

\def\push{\texttt{PUSH}\xspace}
\def\pull{\texttt{PULL}\xspace}
\def\ex{\texttt{EXCHANGE}\xspace}

\def\S{\mathcal{S}\xspace}

\def\whp{\emph{w.h.p. }}

\newtheorem{theorem}{Theorem}
\newtheorem*{clm:D_in_const_delta}{Claim \ref{clm:D_in_const_delta}}
\newtheorem*{lemma:geom_upper_bound}{Lemma \ref{lemma:geom_upper_bound}}
\newtheorem*{thm:tree_of_queues}{Theorem \ref{thm:tree_of_queues}}
\newtheorem*{lemma:algebraic_gossip_in_tree_single_node}{Lemma \ref{lemma:algebraic_gossip_in_tree_single_node}}
\newtheorem*{lemma:sum_of_degrees_on_shortest_path}{Lemma \ref{lemma:sum_of_degrees_on_shortest_path}}
\newtheorem*{thm:linear_broadcast}{Theorem \ref{thm:linear_broadcast}}
\newtheorem*{theorem:async-info-spr}{Theorem \ref{theorem:async-info-spr}}

\newtheorem{lemma}{Lemma}
\newtheorem{corollary}{Corollary}
\newtheorem{dfn}{Definition}
\newtheorem{clm}{Claim}

\newcommand{\Rmnum}[1]{\expandafter\@slowromancap\romannumeral #1@}

\pdfpageattr{/Group <</S /Transparency /I true /CS /DeviceRGB>>}

\begin{document}
\begin{titlepage}

\title{Order Optimal Information Spreading Using Algebraic Gossip}

\author{
Chen Avin\footnotemark[1]
\and
Michael Borokhovich\footnotemark[1]
\and
Keren Censor-Hillel\footnotemark[2]
\and
Zvi Lotker\footnotemark[1]
}

\def\thefootnote{\fnsymbol{footnote}}

\footnotetext[1]{
\noindent
Department of Communication Systems Engineering, Ben Gurion University,
Beer-Sheva, Israel.
E-mail:{\tt \{avin,borokhom,zvilo\}@cse.bgu.ac.il}.
}

\footnotetext[2]{
\noindent
Computer Science and Artificial Intelligence Laboratory, MIT.
E-mail:{\tt ckeren@csail.mit.edu}.
Supported by the Simons Postdoctoral Fellows Program.
}

\date{}
\maketitle

\thispagestyle{empty}

%%%%%%%%%%%%%%%%%%%%%%%%%%%%%%%%%%%%%%%%%%%%%%%%%%%%%%%%%%%%%%%%

\begin{abstract}
%\michael{Abstract needs to be updated to include Keren's results. The bound for asynchronous time requires $d(\S)$.}
%In this paper we continue recent studies on
In this paper we study gossip based information spreading with bounded message sizes. We use algebraic gossip to disseminate $k$ distinct messages to all $n$ nodes in a network. For arbitrary networks we provide a new upper bound for uniform algebraic gossip of $O((k+\log n + D)\Delta)$ rounds with high probability, where $D$ and $\Delta$ are the diameter and the maximum degree in the network, respectively.
For many topologies and selections of $k$ this bound improves previous results, in particular,
for graphs with a constant maximum degree it implies that uniform gossip is \emph{order optimal} and the stopping time is $\Theta(k + D)$.

To eliminate the factor of $\Delta$ from the upper bound we propose a non-uniform gossip protocol, TAG,  which is based on algebraic gossip and an arbitrary spanning tree protocol $\S$. The stopping time of TAG is $O(k+\log n +d(\S)+t(\S))$, where $t(\S)$ is the stopping time of the spanning tree protocol, and $d(\S)$ is the diameter of the spanning tree. We provide two general cases in which this bound leads to an order optimal protocol. The first is for $k=\Omega(n)$, where, using a simple gossip broadcast protocol that creates a spanning tree in at most linear time, we show that TAG finishes after $\Theta(n)$ rounds for any graph.
The second uses a sophisticated, recent gossip protocol to build a fast spanning tree on graphs with large weak conductance. In turn, this leads to the optimally of TAG on these graphs for $k=\Omega(\mathrm{polylog}(n))$.
%Finally, we propose a simple gossip broadcast algorithm that creates a spanning tree in at most linear time. The last implies that TAG is an order optimal gossip dissemination protocol for the case $k=\Omega(n)$, i.e., it finishes after $\Theta(n)$ rounds for any graph.  For certain graphs this is an improvement of order $n$ from the uniform gossip.
The technique used in our proofs relies on queuing theory, which is an interesting approach that can be useful in future gossip analysis.
\end{abstract}

\vspace{6cm}
\small
{Michael Borokhovich is a full-time student at Ben Gurion University and is principally responsible for the paper's contributions.}

\end{titlepage}
%%%%%%%%%%%%%%%%%%%%%%%%%%%%%%%%%%%%%%%%%%%%%%%%%%%%%%%%%%%%%%%%

\section{Introduction}\label{sec:introduction}
%What we need:
%1. problem statement (all-to-all, k-dissimination) motivation
% Gossip
%2. network coding RLNC
%3. uniform, non-uniform,
%4. time model
%5. problems, some recent work....

% ---------- Motivation ------------
One of the most basic information spreading
applications is that of disseminating information stored at a subset of
source nodes to a set of sink nodes.  Here we consider the {\em $k$-dissemination}
case: $k$ initial messages ($k\le n$)  located at some nodes (a node can hold more than one initial message) need to
reach all $n$ nodes. The {\em all-to-all communication} -- each of $n$ nodes has an initial value that is needed to be
disseminated to all nodes -- is a special case of $k$-dissemination. The goal is to perform this task in the lowest possible number of time steps when messages have \emph{limited} size (i.e., a node may not be able to send all its data in one message).
%Broadcast, multicast, and anycast are all variants of the above applications.

% ---------- Gossip --------
Gossiping, or rumor-spreading, is a simple stochastic process for
dissemination of information across a network. In a synchronous \emph{round} of gossip,
\emph{each} node chooses a \emph{single} neighbor as the \emph{communication
  partner} and takes an action. In an asynchronous time model a single node wake-ups and chooses the \emph{communication
  partner} and $n$ consecutive steps are considered as one \emph{round}. The \emph{gossip communication model} defines how to select this neighbor, e.g., \emph{uniform} gossip
  is when the communication
  partner is selected uniformly at random from the set of all neighbors.  We then consider three possible actions:
either the node pushes information to the partner (\push), pulls
information from the partner (\pull), or does both (\ex), but here we mostly present results about \ex.

% ------ Protocols ---------------
A \emph{gossip protocol} uses a gossip communication model in conjunction with the choice of the particular content that is exchanged.
Due to their distributed nature, gossip protocols have gained popularity in recent years and have found
applications both in communication networks (for example, updating
database replicated at many
sites~\cite{demers88epidemic,Karp2000Randomized}, computation of
aggregate information~\cite{Kempe2003Gossip} and multicast via network
coding~\cite{Deb2006Algebraic}, to name a few) as well as in social
networks~\cite{Kempe2003Maximizing,Chaintreau2008Opportunistic}.

% --------- RLNC ------------
In the current work we analyze \emph{algebraic gossip} which is a type of network coding known as random linear coding (RLNC) \cite{Medard2002Beyond,LiYeuCai03} that uses gossip algorithms for all-to-all
  communication and $k$-dissemination. In algebraic gossip the content of messages is the random linear combination of all messages stored at a sender. Once a node has received enough independent messages (independent linear equations) it can solve the system of linear equations and discover all the initial values of all other nodes. It has been proved \cite{ho03the-benefits} that network coding can improve the throughput of the network by better sharing of the network resources. Note, however, that in gossip protocols, nodes select a single partner, so for $k$-dissemination to succeed each node needs to receive at least $k$ messages (of bounded size), hence at least a total of $kn$ messages need to be sent and received. This immediately leads to a trivial lower bound of $\Omega(k)$ rounds for $k$-dissemination.

We study uniform and non-uniform algebraic gossip both in the synchronous and the asynchronous time models on arbitrary graph topologies. The stopping time obviously depends on the protocol, the gossip communication model, the graph topology, but also on the time model, as sown in other cases \cite{Georgiou2008On-the-complexity}. We now give an overview of our results followed by a discussion of previous work.

%We look at the following scenario: a connected network with $n$ nodes and $k$ ($k\le n$) initial messages located at some nodes (a node can hold more than one initial message). The goal is to disseminate all the $k$ messages to all the $n$ nodes in the lowest possible number of time steps. We consider only gossip protocols, i.e., at every transmission opportunity, a node sends a single limited-size message to a single neighbor chosen from the set (or subset) of the node's neighbors. Moreover, in gossip protocols there is no managing entity, thus all node's decisions (\emph{what} to send and \emph{to whom}) are local. The last implies the distributed nature of gossip protocols making them suitable for sensors, ad hock, or p2p networks. Variations of gossip protocols for information dissemination and calculation of aggregate functions were proposed and analyzed in \cite{KarpEtAl00a,Kempe2003Gossip,Boyd2006Randomized}. In the current work we analyze the algebraic gossip protocol which is an information dissemination protocol based on random linear network coding (RLNC) \cite{Medard2002Beyond,LiYeuCai03}.

\subsection{Overview of Our Results}
%Need:
%1. continue of previous results and approach and related work pointer
%2. uniform AG
%3. Tree based
%4. Round robin
%5. Keren broadcast

Our first set of results is about the stopping time of uniform algebraic gossip. In \cite{Borokhovich2010Tight} we have shown a tight bound of $\Theta(n)$ for all-to-all communication for graphs with constant maximum degree. To prove this, we used a reduction of gossip to a network of queues and analyzed the waiting times in the queues. Bounding the general $k$-dissemination case is significantly harder, despite some similarity in the tools used.
%Recently, Haeupler \cite{Haeupler2010Analyzing}, used conductance-based methods to drive some tight bounds for the $k$-dissemination, we'll discuss the differences in the next section.
Unless explicitly stated, all our results are for gossip using \ex and are with high probability\footnote{An event occurs with high probability (\whp) if its probability is of at least $1-O(\tfrac{1}{n}$).}.

We provide a novel upper bound for uniform algebraic gossip of $O((k+\log n + D)\Delta)$  where $D$ is the diameter and $\Delta$ is the maximum degree in the graph.
For graphs with constant maximum degree this leads to a bound of  $O(k + D)$. For the synchronous case we have a matching lower bound of $\Omega(k + D)$ which makes uniform algebraic gossip an order \textbf{optimal} gossip protocol for these graphs.
We conjecture that the optimality holds for the asynchronous time model as well, but only show it when $k=\Omega(D)$.

However, there are topologies for which uniform algebraic gossip performs badly, e.g., in the barbell graph (two cliques connected with a single edge) it takes $\Omega(n^2)$ rounds to perform all-to-all communication \cite{Borokhovich2010Tight}. This is usually the result of bottlenecks that exist in the graph and lead to low conductance. For such "bad" topologies we propose here a modification of the uniform algebraic gossip called \emph{Tree based Algebraic Gossip} (TAG).
The basic idea of the protocol is that it operates in two phases: first, using a gossip protocol $\S$ it generates a spanning tree in which  each node in the tree has a single parent. In the next phase, algebraic gossip is performed on the tree where each node does \ex with its parent.
Let $t(\S)$ and $d(\S)$ be the stopping time of $\S$ and the diameter of the tree generated by $\S$, respectively.
For any spanning tree gossip protocol $\S$ we prove for TAG an upper bound of: $O(k+\log n +d(\S)+ t(\S))$ for the \emph{synchronous} and the \emph{asynchronous} time models.
%, and a bound of $O(k+\log n + d(\S)+t(\S))$ for the \emph{asynchronous} time model.
As a special case of a spanning tree protocol, one can use a gossip broadcast (or $1$-dissemination) protocol $\mathcal{B}$ -- a protocol in which a single message originated at some node should be disseminated to all nodes. Interestingly, using a gossip broadcast for the spanning tree construction in TAG, eliminates the dependence on the diameter of the spanning tree in the synchronous time model, i.e., if we use $\mathcal{B}$ as $\S$, we obtain the bound of $O(k+\log n +t(\mathcal{B}))$ rounds.
For a general spanning tree protocol $\S$,
it follows directly that if $k=\Omega(\max(\log n, d(\S), t(\S)))$, TAG is an order \textbf{optimal} with a stopping time of $\Theta(k)$.
We provide two examples of this scenario: 
the first example leads to the most significant result of the paper. Using a simple round-robin-based broadcast we show that TAG is an order optimal gossip protocol for {\em $k$-dissemination} in any topology when $k=\Omega(n)$. This imply, somewhat surprisingly, that for \textbf{any graph}, if $k=\Omega(n)$, TAG finishes in $\Theta(n)$ rounds. In the barbell graph mentioned above, TAG leads to a speedup ratio of $n$ compare to the uniform algebraic gossip. The second example makes use of a recent non-uniform information dissemination protocol from \cite{censor2010fast} that works well on graphs $G$ with large \emph{weak conductance} denoted by $\Phi_c(G)$ for a parameter $c$ (see Section \ref{sec:weak}). We provide sufficient conditions on $k$, $c$ and $\Phi_c(G)$ that make TAG order optimal when using the protocol of \cite{censor2010fast} as a spanning tree protocol. Table \ref{tab:results} summarizes our main results of the paper and next, we discuss previous results.

\renewcommand{\arraystretch}{1.2}
\begin{table}[t]\label{tab:results}
	\centering
		\begin{tabular}{|c| c|| c | c |}
		\hline
			Protocol & Graph & Synchronous & Asynchronous\\
			\hline \hline
			\multirow{2}{*}{Uniform AG} & any graph & 	\multicolumn{2}{c|}{$O((k+\log n + D)\Delta)$} \\
			\cline{2-4}
			 & constant max degree  & $\mathbf{\Theta(k+D)}$ &$O(k + D)$ (*) \\
			\hline
			\hline
			\multirow{4}{*}{TAG}
			& \multirow{2}{*}{any graph}  &\multicolumn{2}{c|}{$O(k+\log n + d(\S)+t(\S))$}\\		
			\cline{3-4}
			&  &$O(k+\log n +t(\mathcal{B}))$ & $O(k+\log n +d(\mathcal{B})+t(\mathcal{B}))$\\
			\cline{2-4}
			& $k=\Omega(n)$, any graph& \multicolumn{2}{c|}{ $\mathbf{\Theta(n)}$}\\
			\cline{2-4}
			& $c=O(\log^p{(n)})$  & \multirow{2}{*}{$\mathbf{\Theta(k)}$} & \multirow{2}{*}{$O(k +  d(IS))$ (**)}\\		
			& $k=\Omega(\log^{2p+3}{(n)})$ & & \\
			\hline
		\end{tabular}
	\caption{Overview of the main results of the paper. \textbf{Bold text} and $\Theta$ indicate order optimal result. (*) we prove an upper bound but conjecture it is optimal. (**) we prove the upper bound but conjecture is should be $\Theta(k)$. $\S$ is a spanning tree protocol, $\mathcal{B}$ is a broadcast protocol, and IS is an information dissemination gossip protocol from \cite{censor2010fast}.}
	\label{tab:OurResultsUniform}
\end{table}

%\begin{table}[htbp]
%	\centering
%		\begin{tabular}{|c|| c| c|}
%		\hline
%			Graph & Synchronous & Asynchronous\\
%			\hline \hline
%			Any Graph & $O(k+\log n + t(\S))$ & $O(k+\log n + d(\S)+t(\S))$  \\\hline
%			Any Graph with $k=n$& \multicolumn{2}{c|}{ $\Theta(n)$ using $\S=\mathcal{B_{RR}}$} \\
%%			\hline
%%			Constant Max Degree  & $\Theta(k + t(\S))$  & $O(k+\log n + d(\S)+t(\S))$ \\
%			\hline
%		\end{tabular}
%	\caption{Our results for the new algebraic gossip protocol TAG}
%	\label{tab:OurResultsNewProtocol}
%\end{table}

\subsection{Related Work}\label{sec:related_work}

Uniform algebraic gossip was first proposed by Deb \emph{et al.} in \cite{Deb2006Algebraic}. The authors studied uniform algebraic gossip using \pull and \push on the \emph{complete graph} and showed a tight bound of $\Theta(k)$, for the case of $k=\omega(log^3(n))$ messages.
Boyd \emph{et al.} \cite{Boyd2006Randomized, Boyd2005Gossip} studied the stopping time of a gossip protocol for the \emph{averaging problem} using the \ex algorithm. They gave a bound for symmetric networks that is based on the second largest eigenvalue of the transition matrix or, equally, the mixing time of a random walk on the network, and showed that the mixing time captures the behavior of the protocol.
Mosk-Aoyama and Shah \cite{Mosk-Aoyama2006Information} used a similar approach to \cite{Boyd2006Randomized, Boyd2005Gossip} to first analyze algebraic gossip on arbitrary networks. They consider symmetric stochastic matrices that (may) lead to a non-uniform gossip and gave an upper bound for the \pull algorithm that is based on a measure of conductance of the network. As the authors mentioned, the offered bound is not tight, which indicates that their conductance-based measure does not capture the full behavior of the protocol.

%Another (\textbf{[Michael]} unpublished) work by Vasudevan and Kudekar \cite{vasudevan-2009} offered the use of \ex together with algebraic gossip. Moreover, the authors give a uniform strong bound on algebraic gossip for arbitrary networks: $O(n\log n)$ in expectation and  $O(n\log^2n)$ with high probability. This result was disproved in \cite{Borokhovich2010Tight}.
In \cite{Borokhovich2010Tight}, we used queuing theory as a novel approach for analyzing algebraic gossip. We then gave an upper bound of $O(n \Delta)$ rounds for any graph for the case of all-to-all communication, where $\Delta$ is the maximum degree in the graph. In addition, a lower bound of $\Omega(n^2)$ was obtained for the barbell graph -- the worst case graph for algebraic gossip.
The bounds (upper and lower) in \cite{Borokhovich2010Tight} were tight in the sense that they matched each other for the worst case scenario. The parameter $\Delta$ is simple and convenient to use, but, it does not fully capture the behavior of algebraic gossip. While it gives optimal ($\Theta(n)$) result for any constant-degree graphs (e.g., line, grid), it fails to reflect the stopping time of algebraic gossip on the complete graph, for example, by giving the $O(n^2)$ bound instead of $O(n)$.

A recent (yet, unpublished) work of Haeupler \cite{Haeupler2010Analyzing} is the most related to our work. Haeupler's paper makes a significant  progress in analyzing the stopping time of algebraic gossip. While all previous works on algebraic gossip used the notion of \emph{helpful message/node} to look at the rank evaluation of the matrices each node maintains (this approach was initially proposed by \cite{Deb2006Algebraic}), Haeupler used a completely different approach. Instead of looking on the growth of the node's subspace (spanned by the linear equations it has), he proposed to look at the orthogonal complement of the subspace and then analyze the process of its disappearing. This elegant and powerful approach led to very impressive results. First, a tight bound of $\Theta(n/\gamma)$ was proposed for all-to-all communication, where $\gamma$ is a min-cut measure of the a related graph. This bound perfectly captures algebraic gossip behavior for any network topology. For the case of $k$-dissemination, the author gives a conjecture that the upper bound is of the form of $O(k+T)$ where $T$ is the time to disseminate a single message to all the nodes. But formally, the bound that is proved is $O(k/\gamma + \log^2 n /\lambda)$ where $\lambda$ is a conductance-based measure of the graph. The work in \cite{Haeupler2010Analyzing} implicitly considered the uniform algebraic gossip, but could be extend to non-uniform cases. It is therefore hard to compare TAG to the results of \cite{Haeupler2010Analyzing}, nevertheless, our bounds for the uniform algebraic gossip are better for certain families of graphs. Table \ref{tab:compare} presents few such examples.

%For example, the bound of $O(k/\gamma + \log^2 n /\lambda)$ timeslots given in \cite{Haeupler2010Analyzing} (Lemma 7.6) yields the following stopping times: line graph -- $O(k+n\log^2 n)$ rounds, grid graph - $O(k+\sqrt{n}\log^2 n)$, binary tree graph -- $O(k+n\log^2 n)$. For all the these graphs, our bound yields the stopping time of $O(k + D)$ rounds. In the case of the line and grid graphs, the improvement is by a factor of $\log^2 n$, and in the case of a binary tree graph, the improvement is by $n\log n$ for a constant $k$, and by $\sqrt{n}\log^2 n$ for $k=\sqrt{n}$.

\begin{table}[htbp]\label{tab:compare}
	\centering
		\begin{tabular}{|c||c|c| m{4cm} |}
		\hline
			Graph & $O(k/\gamma + \log^2 n /\lambda)/n$ \cite{Haeupler2010Analyzing} & $O((k+\log n + D)\Delta)$ [here] & Improvement factor\\
			\hline \hline
			Line & $O(k+n\log^2 n)$ & $O(k + n)$ & $\log^2 n$\\\hline
			Grid &$O(k+\sqrt{n}\log^2 n)$ &$O(k + \sqrt{n})$ & $\log^2 n$ for $k=O(\sqrt{n})$ \\\hline
			Binary Tree &$O(k+n\log^2 n)$ &$O(k + \log n)$ &$\Omega(\frac{n\log n}{k})$\\
			\hline
		\end{tabular}
	\caption{Comparison of our results with \cite{Haeupler2010Analyzing}}
	\label{tab:ComparisonOfOurResultsWithCite}
\end{table}

To give a quick summary of our results and previous work, the two main contributions of the paper are i) we prove that for graphs with constant maximum degree uniform algebraic gossip is order optimal for $k$-dissemination in the synchronous time model and ii) we offer a new non-uniform
algebraic gossip protocol, TAG, that is order optimal for large selections of graphs and $k$. The rest of the paper is organized as follows: in Section \ref{sec:pre} we give definitions. Section \ref{sec:uniform} proves results for uniform algebraic gossip and Section \ref{sec:tag} presents the TAG protocol and its general bound. Sections \ref{sec:rr} and \ref{sec:weak}, then, discuss cases where TAG is optimal.

\section{Preliminaries}\label{sec:pre}

%\subsection{Network and Time Models}
We model the communication network by a connected undirected graph $G_n = G_n(V,E)$, where $V$ is the set of vertices and $E$ is the set of edges. Number of vertices in the graph is $\left|V\right| = n$.
Let $N(v)\subseteq V$ be a set of neighbors of node $v$ and $d_{v}=\lvert N(v)\rvert$ its degree, let $\Delta = \max_v d_v$ be the maximum degree of $G_n$, and let $D$ be the diameter of the graph.

We consider two time models: asynchronous and synchronous.
In the \emph{asynchronous} time model at every \textbf{timeslot}, \textbf{one node} selected independently and uniformly at random, takes an action and a single pair of nodes communicates\footnote{Alternatively, this model can be seen as each node having a clock which ticks at the times of a rate 1 Poisson process and there is a total $n$ clock ticks per round \cite{Boyd2006Randomized}.}. We consider $n$ consecutive timeslots as one \emph{round}. In the \emph{synchronous} time model at every \textbf{round}, \textbf{every node} takes an action and selects a single communication partner. It is assumed that the information received in the current round will be available to a node for sending only at the beginning of the next round.
%\subsection{Gossip Communication Models} % (Gossip Algorithms)
A \textbf{Gossip communication model} (sometimes called gossip algorithm) defines the way information is spread in the network.
In the gossip communication model, a node that wakes up (according to the time model) can initiate communication only with a single neighbor\footnote{Note that this implies that in the synchronous model a node can communicate with more than a single neighbor, if other nodes initiate communication with it.} (i.e., communication partner). The model describes how the communication partner is chosen and in which direction (to -- \push, from --\pull, or both -- \ex) the message is sent. In this work we use the following communication models:
\begin{dfn}[Uniform Gossip]
Uniform gossip is a gossip in which a communication partner is chosen randomly and uniformly among all the neighbors.
\end{dfn}

\begin{dfn}[Round-Robin ($\mathcal{RR}$) Gossip]
\label{dfn:round-robin}
In round-robin gossip, the communication partner is chosen according to a fixed, cyclic list, of the nodes' neighbors.
This list dictates the order in which neighbors are being contacted.
If the initial partner is chosen at random, this gossip communication model is known as the \emph{quasirandom rumor spreading model}\cite{ADHP2009,Doerr08quasirandomrumor}.
\end{dfn}

%As mention earlier, after the partner selection a node can either \push, \pull, or \ex message with their partner.

%\begin{dfn}[\push /\pull /\ex Gossip]
%In \push gossip, one message is sent to the communication partner.
%In \pull gossip, one message is sent from the communication partner.
%In \ex gossip, one message is sent to the communication partner, and one message is sent from the communication partner.
%\end{dfn}

%\begin{dfn}[Maximum Gossip Degree]
%Maximum gossip degree is defined as $\Delta_{gos}=\max_{v\in V}d^{gos}_v$, where $d^{gos}_v$ is the number of neighbors from which the node $v$ chooses a communication partner in the gossip algorithm.
%\end{dfn}

%\subsection{Gossip Protocols}

\paragraph{Gossip Protocols}define the task and  the message content. In turn, a gossip protocol can use any of the gossip communication models defined above (and others).
We will use two types of gossip protocols here. The first is \textbf{STP Gossip} -- protocols whose task is to create a \emph{spanning tree} of the graph. The goal of a Gossip STP protocol $\S$ is that every node, except a node which is the \emph{root}, will have a single neighbor called the \emph{parent}. Note that one simple way to generate a spanning tree is by using a $1$-dissemination protocol, namely a broadcast protocol.

The second protocol, is a $k$-dissemination protocol called \textbf{Algebraic Gossip}.
In algebraic gossip, every message sent by a node is sent according to the random linear coding (RLNC) technique which is described next. As mentioned, algebraic gossip can use any of the communication models presented above.

\paragraph{Random Linear Network Coding (RLNC)}\label{subsec:rand_lin_coding}
The random linear network coding approach is used in algebraic gossip for building outgoing messages to achieve fast information dissemination.
Let $\mathbb{F}_q$ be a field of size $q$. There are $k\le n$ initial messages ($x_1,...,x_k$) that are represented as vectors in $\mathbb{F}_q^r$. We can represent every message as an integer value bounded by $M$, and therefore, $r=\left\lceil \log_q(M)\right\rceil$. All transmitted messages have a fixed length and represent linear equations over $\mathbb{F}_q$. The variables (unknowns) of these equations are the initial values $x_i\in\mathbb{F}_q^r, \:\: 1 \le i \le k$ and a message contains the coefficients of the  variables  and the result of the equation; therefore the length of each message is: $r\log_2q+k\log_2q$ bits (and it is usually assumed that $r \gg n$ \cite{Deb2006Algebraic}).  A message is built as a random linear combination of all messages stored by the node and the coefficients are drawn uniformly at random from $\mathbb{F}_q$.
A received message will be appended to the node's stored messages only if it is independent of all linear equations (messages) that are already stored by the node and otherwise it is ignored. Nodes store messages (linear equations) in a matrix form and once the dimension (or rank) of the matrix becomes $k$, a node can solve the linear system and discover all the $k$ messages.

The following definition is necessary for understanding the concept of helpfulness in the analysis of algebraic gossip.
\begin{dfn}[Helpful node and helpful message]
\label{dfn:helpful_node_and_message}
We say that a node $x$ is a \textbf{\emph{helpful node}} to a node $y$ if and only if a random linear combination constructed by $x$ can be linearly independent with all equations (messages) stored in $y$.
We call a message a \textbf{\emph{helpful message}} if it increases the dimension (or rank) of the node (i.e., the rank of the matrix in which the node stores the messages).
\end{dfn}

\section{\texorpdfstring{$k$}{k}-dissemination with Uniform Algebraic Gossip}\label{sec:uniform}

The main result of this section is that uniform algebraic gossip is order optimal $k$-dissemination for graphs with constant maximum degree and for any selection of $k$. It is formally stated in Theorem \ref{thm:constantmax} and is an almost direct result of the following general bound for uniform algebraic gossip:

\begin{theorem}
\label{thm:algebraic_gossip_with_k}
For any connected graph $G_n$, the stopping time of the uniform algebraic gossip protocol with $k$ messages is $O((k+\log n + D)\Delta)$ rounds for synchronous and asynchronous time models \whp %, with probability of at least $1-O(\tfrac{1}{n})$.
\end{theorem}
%\begin{proof}
The idea of the proof relies on the queuing networks technique we presented in \cite{Borokhovich2010Tight}. The major steps of the proof are:
\begin{itemize}
\item Perform a Breath First Search (BFS) on $G_n$ starting at an arbitrary node $v$. The search results in a directed shortest path spanning tree $T_n$ rooted at $v$. The maximum depth $l_{\max}$ of the tree $T_n$ rooted at $v$ is at most $D$.
\item Reduce the problem of algebraic gossip on a tree $T_n$ to a simple system of queues $Q_n^{tree}$ rooted at $v$, where at each node we assume an infinite queue with a single server. Every initial message becomes a customer in the queuing system. The root $v$ finishes once all the customers arrive at it.
%\item Show that the maximum depth $d(\S)$ of the tree $T_n$ rooted at $v$ is at most $D$.
\item Show that the stopping time of the tree topology queuing system -- $Q_n^{tree}$, is $O((k+\log n+l_{\max})n\Delta)$ timeslots \whp So, we obtain the stopping time for the node $v$.
\item Use union bound to obtain the result for all the nodes in $G_n$.
\end{itemize}

Just before we start the formal proof of Theorem \ref{thm:algebraic_gossip_with_k}, we present an interesting theorem related to queuing theory. The theorem gives the stopping time of the feedforward queuing system \cite{Chen2001Fundamentals} arranged in a tree topology. Consider the following scenario: $n$ identical M/M/1 queues arranged in a tree topology. There are no external arrivals, and there are $k$ customers arbitrarily distributed in the system. In the feedforward network, a customer can not enter the same queue more than once, thus, customers eventually leave the system via the queue at the root of the tree. We ask the following question: how much time will it take for the last customer to leave the system?

\begin{theorem}
\label{thm:tree_of_queues}
Let $Q_n^{tree}$ be a network of $n$ nodes arranged in a tree topology, rooted at the node $v$. The depth of the tree is $l_{\max}$. Each node has an infinite queue, and a single exponential server with parameter $\mu$. The total amount of customers in the system is $k$ and they are initially distributed arbitrarily in the network. The time by which all the customers leave the network via the root node $v$ is $t({Q}_n^{tree})=O((k+l_{\max}+\log n)/\mu)$ timeslots with probability of at least $1-\tfrac{2}{n^2}$.
\end{theorem}

The main idea of the proof is to show that the stopping time of the network $Q_n^{tree}$ (i.e., the time by which all the customers leave the network) is stochastically\footnote{For completeness, stochastic dominance is formally defined in appendix.} smaller or equal to the stopping time of the systems of $l_{\max}$ queues arranged in a line topology -- $Q_{l_{\max}}^{line}$. Then, we make the system $Q_{l_{\max}}^{line}$ stochastically slower by moving all the customers out of the system and make them enter back via the farthest queue with the rate $\lambda=\mu/2$. Finally, we use Jackson's Theorem for open networks to find the stopping time of the system. See Fig. \ref{fig:reduction_to_queues} for the illustration. The full proof of the above theorem can be found in the appendix. We can now prove Theorem  \ref{thm:algebraic_gossip_with_k}.

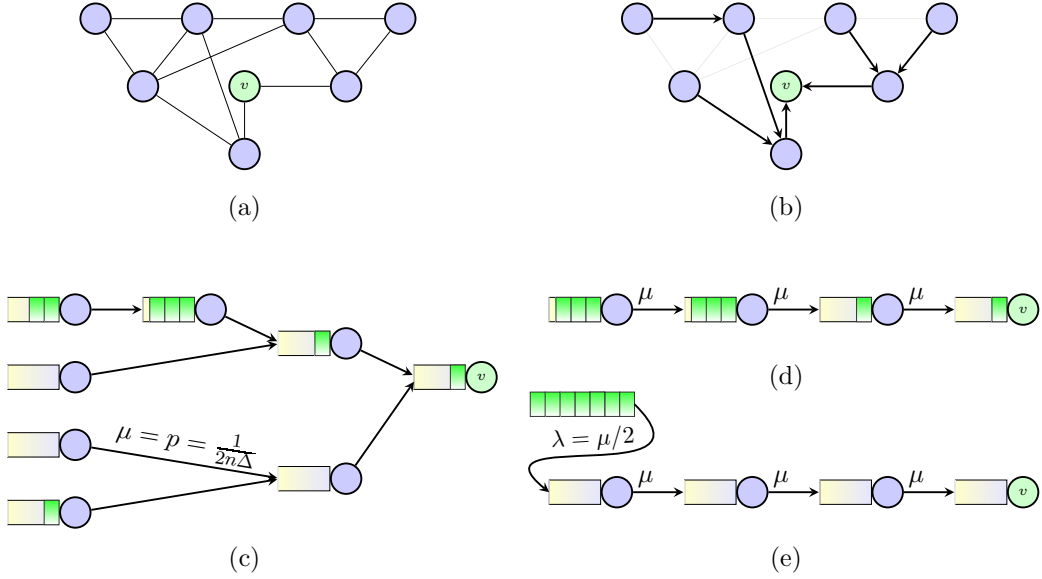
\begin{figure}[t]
\centering
\scalebox{0.9}{\begin{tikzpicture}
[inner sep=0.6mm, place/.style={circle,draw=black,fill=blue!20,thick,minimum size=4.5mm},>=stealth,place1/.style={circle,draw=black,fill=green!20,thick,minimum size=4.5mm}]

\def\myX{0.3}

\node at (0+\myX,0) [place] (a) {\tiny{$$}};
\node at (1.5+\myX,0) [place] (b) {\tiny{$$}};
\node at (0.7+\myX,-1) [place] (c)  {\tiny{$$}};
\node at (2.2+\myX,-1) [place1] (d) {\tiny{$v$}};
\node at (3+\myX,0) [place] (e)  {\tiny{$$}};
\node at (3.7+\myX,-1) [place] (f)  {\tiny{$$}};
\node at (2.2+\myX,-2) [place] (g)  {\tiny{$$}};
\node at (4.5+\myX,0) [place] (h)  {\tiny{$$}};

\path (a) edge (b);
\path (a) edge (c);
\path (b) edge (c);
\path (d) edge (f);
\path (h) edge (f);
\path (e) edge (c);
\path (d) edge (g);
\path (c) edge (g);
\path (b) edge (g);
\path (b) edge (e);
\path (h) edge (e);
\path (e) edge (f);

\node at (\myX+2.2, -2.8) [auto]{(a)};

\def\myX{8.3}

\node at (0+\myX,0) [place] (a) {\tiny{$$}};
\node at (1.5+\myX,0) [place] (b) {\tiny{$$}};
\node at (0.7+\myX,-1) [place] (c)  {\tiny{$$}};
\node at (2.2+\myX,-1) [place1] (d) {\tiny{$v$}};
\node at (3+\myX,0) [place] (e)  {\tiny{$$}};
\node at (3.7+\myX,-1) [place] (f)  {\tiny{$$}};
\node at (2.2+\myX,-2) [place] (g)  {\tiny{$$}};
\node at (4.5+\myX,0) [place] (h)  {\tiny{$$}};

\path[->,black,thick] (a) edge (b);
\path[opacity=0.09] (a) edge (c);
\path[opacity=0.09] (b) edge (c);
\path[<-,black,thick] (d) edge (f);
\path[->,black,thick] (h) edge (f);
\path[->,black,thick] (e) edge (f);
\path[opacity=0.09] (e) edge (c);
\path[<-,black,thick] (d) edge (g);
\path[->,black,thick] (c) edge (g);
\path[->,black,thick] (b) edge (g);
\path[opacity=0.09] (b) edge (e);
\path[opacity=0.09] (h) edge (e);

\node at (\myX+2.2, -2.8) [auto]{(b)};

\def\myX{0}
\def\myY{-4.3}

\foreach \x /\y/\thead/\ttail/\tname in {0/0/a/a2/v_1,2/0/b/b2/v_2,0/-1/c/c2/v_3,6/-1/v/v2/v,0/-2/e/e2/v_5,4/-2.5/f/f2/v_6,4/-0.5/g/g2/v_7,0/-3/h/h2/v_8}{
\node at (\x-0.95+\myX,\y+\myY) [auto] (\ttail)  {\tiny{}};
\shade[left color=yellow!20,right color=blue!10,draw=black] (\x-1+\myX,\y-0.18+\myY) rectangle (\x-0.24+\myX,\y+0.18+\myY);
\draw[white] (\x-1+\myX,\y+0.18+\myY) -- (\x-1+\myX,\y-0.18+\myY);
\ifthenelse{\equal{\thead}{a}}
{
\shade[top color=green!80,draw=black!80] (\x-0.46+\myX,\y-0.18+\myY) rectangle (\x-0.24+\myX,\y+0.18+\myY);
\shade[top color=green!80,draw=black!80] (\x-0.46+\myX-0.22,\y-0.18+\myY) rectangle (\x-0.24+\myX-0.22,\y+0.18+\myY);
}

\ifthenelse{\equal{\thead}{b}}
{
\shade[top color=green!80,draw=black!80] (\x-0.46+\myX,\y-0.18+\myY) rectangle (\x-0.24+\myX,\y+0.18+\myY);
\shade[top color=green!80,draw=black!80] (\x-0.46+\myX-0.22,\y-0.18+\myY) rectangle (\x-0.24+\myX-0.22,\y+0.18+\myY);
\shade[top color=green!80,draw=black!80] (\x-0.46+\myX-0.22-0.22,\y-0.18+\myY) rectangle (\x-0.24+\myX-0.22-0.22,\y+0.18+\myY);
}

\ifthenelse{\equal{\thead}{c}}
{
}

\ifthenelse{\equal{\thead}{v}}
{
\shade[top color=green!80,draw=black!80] (\x-0.46+\myX,\y-0.18+\myY) rectangle (\x-0.24+\myX,\y+0.18+\myY);
}

\ifthenelse{\equal{\thead}{e}}
{
}

\ifthenelse{\equal{\thead}{f}}
{
}

\ifthenelse{\equal{\thead}{g}}
{
\shade[top color=green!80,draw=black!80] (\x-0.46+\myX,\y-0.18+\myY) rectangle (\x-0.24+\myX,\y+0.18+\myY);
}

\ifthenelse{\equal{\thead}{h}}
{
\shade[top color=green!80,draw=black!80] (\x-0.46+\myX,\y-0.18+\myY) rectangle (\x-0.24+\myX,\y+0.18+\myY);
}

\ifthenelse{\equal{\thead}{v}}{\node at (\x+\myX,\y+\myY) [place1] (\thead)  {\tiny{$\tname$}};}
{\node at (\x+\myX,\y+\myY) [place] (\thead)  {};}
}

\path[->,black,thick] (a) edge (b2);
\path[->,black,thick] (c) edge (g2);
\path[->,black,thick] (b) edge (g2);
\path[->,black,thick] (e) edge node [above,sloped,black] {{$\mu=p=\tfrac{1}{2n\Delta}$}} (f2);
\path[->,black,thick] (h) edge (f2);
\path[->,black,thick] (g) edge (v2);
\path[->,black,thick] (f) edge (v2);

\node at (\myX+2.5, \myY-3.7) [auto]{(c)};

\def\myX{8}

\foreach \x /\y/\thead/\ttail/\tname in {0/0/a/a2/v_1,2/0/b/b2/v_2,6/0/v/v2/v,4/0/g/g2/v_7}{
\node at (\x-0.95+\myX,\y+\myY) [auto] (\ttail)  {\tiny{}};
\shade[left color=yellow!20,right color=blue!10,draw=black] (\x-1+\myX,\y-0.18+\myY) rectangle (\x-0.24+\myX,\y+0.18+\myY);
\draw[white] (\x-1+\myX,\y+0.18+\myY) -- (\x-1+\myX,\y-0.18+\myY);
\ifthenelse{\equal{\thead}{a}}
{
\shade[top color=green!80,draw=black!80] (\x-0.46+\myX,\y-0.18+\myY) rectangle (\x-0.24+\myX,\y+0.18+\myY);
\shade[top color=green!80,draw=black!80] (\x-0.46+\myX-0.22,\y-0.18+\myY) rectangle (\x-0.24+\myX-0.22,\y+0.18+\myY);
\shade[top color=green!80,draw=black!80] (\x-0.46+\myX-0.22-0.22,\y-0.18+\myY) rectangle (\x-0.24+\myX-0.22-0.22,\y+0.18+\myY);
}

\ifthenelse{\equal{\thead}{b}}
{
\shade[top color=green!80,draw=black!80] (\x-0.46+\myX,\y-0.18+\myY) rectangle (\x-0.24+\myX,\y+0.18+\myY);
\shade[top color=green!80,draw=black!80] (\x-0.46+\myX-0.22,\y-0.18+\myY) rectangle (\x-0.24+\myX-0.22,\y+0.18+\myY);
\shade[top color=green!80,draw=black!80] (\x-0.46+\myX-0.22-0.22,\y-0.18+\myY) rectangle (\x-0.24+\myX-0.22-0.22,\y+0.18+\myY);
}

\ifthenelse{\equal{\thead}{v}}
{
\shade[top color=green!80,draw=black!80] (\x-0.46+\myX,\y-0.18+\myY) rectangle (\x-0.24+\myX,\y+0.18+\myY);
}

\ifthenelse{\equal{\thead}{g}}
{
\shade[top color=green!80,draw=black!80] (\x-0.46+\myX,\y-0.18+\myY) rectangle (\x-0.24+\myX,\y+0.18+\myY);
}

\ifthenelse{\equal{\thead}{v}}{\node at (\x+\myX,\y+\myY) [place1] (\thead)  {\tiny{$\tname$}};}
{\node at (\x+\myX,\y+\myY) [place] (\thead)  {};}
}

\path[->,black,thick] (a) edge node [above,sloped,black, near start] {{$\mu$}}(b2);
\path[->,black,thick] (b) edge node [above,sloped,black, near start] {{$\mu$}}(g2);
\path[->,black,thick] (g) edge node [above,sloped,black, near start] {{$\mu$}}(v2);

\node at (\myX+2.5, \myY-1) [auto]{(d)};

\def\myY{-7}
\def\myXX{0.5}
\def\myYY{1.3}

\foreach \x /\y/\thead/\ttail/\tname in {0/0/a/a2/v_1,2/0/b/b2/v_2,6/0/v/v2/v,4/0/g/g2/v_7}{
\node at (\x-0.95+\myX,\y+\myY) [auto] (\ttail)  {\tiny{}};
\shade[left color=yellow!20,right color=blue!10,draw=black] (\x-1+\myX,\y-0.18+\myY) rectangle (\x-0.24+\myX,\y+0.18+\myY);
\draw[white] (\x-1+\myX,\y+0.18+\myY) -- (\x-1+\myX,\y-0.18+\myY);
\ifthenelse{\equal{\thead}{a}}
{
\shade[top color=green!80,draw=black!80] (\x-0.46+\myX+\myXX,\y-0.18+\myY+\myYY) rectangle (\x-0.24+\myX+\myXX,\y+0.18+\myY+\myYY);
\shade[top color=green!80,draw=black!80] (\x-0.46+\myX-0.22+\myXX,\y-0.18+\myY+\myYY) rectangle (\x-0.24+\myX-0.22+\myXX,\y+0.18+\myY+\myYY);
\shade[top color=green!80,draw=black!80] (\x-0.46+\myX-0.22-0.22+\myXX,\y-0.18+\myY+\myYY) rectangle (\x-0.24+\myX-0.22-0.22+\myXX,\y+0.18+\myY+\myYY);
\shade[top color=green!80,draw=black!80] (\x-0.46+\myX-0.22-0.22+\myXX,\y-0.18+\myY+\myYY) rectangle (\x-0.24+\myX-0.22-0.22+\myXX,\y+0.18+\myY+\myYY);
\shade[top color=green!80,draw=black!80] (\x-0.46+\myX-0.22-0.22*2+\myXX,\y-0.18+\myY+\myYY) rectangle (\x-0.24+\myX-0.22-0.22*2+\myXX,\y+0.18+\myY+\myYY);
\shade[top color=green!80,draw=black!80] (\x-0.46+\myX-0.22-0.22*3+\myXX,\y-0.18+\myY+\myYY) rectangle (\x-0.24+\myX-0.22-0.22*3+\myXX,\y+0.18+\myY+\myYY);
\shade[top color=green!80,draw=black!80] (\x-0.46+\myX-0.22-0.22*4+\myXX,\y-0.18+\myY+\myYY) rectangle (\x-0.24+\myX-0.22-0.22*4+\myXX,\y+0.18+\myY+\myYY);
\shade[top color=green!80,draw=black!80] (\x-0.46+\myX-0.22-0.22*5+\myXX,\y-0.18+\myY+\myYY) rectangle (\x-0.24+\myX-0.22-0.22*5+\myXX,\y+0.18+\myY+\myYY);
}

\ifthenelse{\equal{\thead}{v}}{\node at (\x+\myX,\y+\myY) [place1] (\thead)  {\tiny{$\tname$}};}
{\node at (\x+\myX,\y+\myY) [place] (\thead)  {};}
}

\path[->,black,thick] (a) edge node [above,sloped,black, near start] {{$\mu$}}(b2);
\path[->,black,thick] (b) edge node [above,sloped,black, near start] {{$\mu$}}(g2);
\path[->,black,thick] (g) edge node [above,sloped,black, near start] {{$\mu$}}(v2);

\draw[->,thick] (-0.46+\myX+\myXX+0.22, \myY+\myYY) .. controls(-0.46+\myX+\myXX+1+0.54, \myY+\myYY-1.3) and (-0.46+\myX+\myXX+0.22-2.5, \myY+\myYY-0.3) .. node [above,sloped,black] {{\small{$\lambda=\mu/2$}}}(-1+\myX, \myY);

\node at (\myX+2.5, \myY-1) [auto]{(e)};

\end{tikzpicture}}
\caption{\small{Reduction of AG to a system of queues. (a) -- Initial graph $G_n$. (b) -- Spanning tree $T_n$. (c) -- System of queues $Q_{n}^{tree}$. (d) -- System of queues $Q_{l_{\max}}^{line}$. Stopping time of $Q_{l_{\max}}^{line}$ is larger than of $Q_{n}^{tree}$. (e)--Taking all customers out of the system and use Jackson theorem for open networks.}}
\label{fig:reduction_to_queues}
\end{figure}

\begin{proof}[Proof of Theorem \ref{thm:algebraic_gossip_with_k}]
We start the analysis of the uniform algebraic gossip with $k$ messages and the asynchronous time model. First, we perform a Breath First Search (BFS) on $G_n$ starting at an arbitrary node $v$. The search results in a directed shortest path spanning tree $T_n$ rooted at $v$. The depth of $T_n$ is $l_{\max}$, and since $T_n$ is the shortest path tree, $l_{\max}\le D$, where $D$ is the diameter of the graph.
On the tree $T_n$, consider a message flow towards the root $v$ from all other nodes. Once $k$ \emph{helpful messages} arrive at $v$, it will reach rank $k$ and finish the algebraic gossip protocol.
We ignore messages that are not sent in the direction of $v$. Ignoring part of messages can only increase the stopping time of the algebraic gossip protocol.

We define a queuing system $Q_n^{tree}$ by assuming an infinite queue with a single server at each node. The root of $Q_n^{tree}$ is the node $v$. Customers of our queuing network are \emph{helpful messages}, i.e., messages that increase the rank of a node they arrive at. This means that every customer arriving at some node increases its rank by 1. When a customer leaves a node, it arrives at the parent node.
The queue length of a node represents a measure of \emph{helpfulness} of the node to its parent, i.e., the number of \emph{helpful messages} it can generate for it.

The service procedure at a node is a transmission of a \emph{helpful message} towards the node $v$ (from a node to its parent). Lemma 2.1 in \cite{Deb2006Algebraic} gives a lower bound for the probability of a message sent by a \emph{helpful node} to be a \emph{helpful message}, which is: $1-\tfrac{1}{q}$. In the uniform gossip communication model, the communication partner of a node is chosen randomly among all the node's neighbors in the original graph $G_n$. The degree of each node in $G_n$ is at most $\Delta$. Thus, in the asynchronous time model, in a given timeslot, a \emph{helpful message} will be sent over the edge in a specific direction with probability of at least $(1-\tfrac{1}{q})/n\Delta$, where $\tfrac{1}{n}$ is the probability that a given node wakes up in a given timeslot, $\tfrac{1}{\Delta}$ is the minimal probability that a specific partner (the parent of the node) will be chosen, and $1-\tfrac{1}{q}$ is the minimal probability that the message will be \emph{helpful}.
Thus, we can consider that the service time in our queuing system is geometrically distributed with parameter $p\ge (1-\tfrac{1}{q})/n\Delta$, and since $q\ge 2$, we can assume the worst case: $p=\tfrac{1}{2n\Delta}$.

Lemma 2 in \cite{Borokhovich2010Tight} shows that we can model the service time of each server as an exponential random variable with parameter $\mu=p$, since in this case, exponential servers are stochastically \emph{slower} than geometric. Such an assumption can only increase the stopping time.

Theorem \ref{thm:tree_of_queues} with $\mu=p$ gives us an upper bound for the stopping time of the node $v$, $t_v=O((k+l_{\max}+\log n)2n\Delta)$ timeslots with probability of at least $1-\tfrac{2}{n^2}$. Since the depth of every BFS tree is bounded by the diameter $D$, using a union bound we obtain the upper bound (in timeslots) for all the nodes in $G_n$:
\begin{align}
\Pr\left(\bigcap_{v\in V} t_v=O((k+\log n + D)2n\Delta)\right)>1-\frac{2}{n}.
\end{align}

Thus we obtain the upper bound for uniform algebraic gossip: $O((k+\log n + D)\Delta)$ rounds.
Next, we show that this bound holds also for the synchronous time model. The proof for the synchronous time model is almost the same as in the asynchronous case, except for the following change.
Instead of dividing time into timeslots, we measure it by rounds ($1$ round = $n$ timeslots). In a given \textbf{round}, a \emph{helpful message} will be sent over the edge in a specific direction with probability $p\ge(1-\tfrac{1}{q})/\Delta$, where the $\tfrac{1}{\Delta}$ is the minimal probability that a specific partner (the parent of the node) will be chosen, and $1-\tfrac{1}{q}$ is the minimal probability that the message will be \emph{helpful}. Since $q\ge 2$, we can assume the worst case: $p=\tfrac{1}{2\Delta}$. The difference from the asynchronous model is the factor of $n$ in $p$, since in the synchronous model, every node wakes up exactly once in a each round. Moreover, in the synchronous case (and in the \ex gossip variation) there is a possibility to receive $2$ messages from the same node in one round (in the asynchronous time model it was impossible to receive $2$ messages from the same node in one timeslot). We assume that if a node receives $2$ messages from the same node at the same round, it will discard the second one. Such an assumption can only increase the stopping time of the protocol, and will make our analysis simpler.
From that point on, the analysis is exactly the same as in the asynchronous case since Theorem \ref{thm:tree_of_queues} does not depend on the time model.
\end{proof}

\subsection{Optimality for Constant Maximum Degree Graphs and Synchronous Time}
Following Theorem \ref{thm:algebraic_gossip_with_k} we can state the main results of the section:
\begin{theorem}
\label{thm:constantmax}
For any connected graph $G_n$ with constant maximum degree, the stopping time of the uniform algebraic gossip protocol
with k messages is $\Theta(k + D)$ in the synchronous time and $O(k + D)$ in the asynchronous time \whp
\end{theorem}

\begin{proof}
To show the upper bound the following simple claim is proved in the appendix:
\begin{clm}
\label{clm:D_in_const_delta}
For any connected graph $G_n$ with a constant maximum degree $(\Delta=O(1))$, the diameter of $G_n$ is $\Omega(\log n)$.
\end{clm}
Now, using Claim \ref{clm:D_in_const_delta} and fact the the maximum degree is constant the upper bound follows.
For the lower bound note that in order to disseminate $k$ messages to $n$ nodes, at least $kn$ transmissions should occur in the network. In synchronous time model, $kn$ transmissions require at least $k/2$ rounds, since every round at most $2n$ messages are sent (2 transmissions per communication pair). In the asynchronous time model, $kn$ transmissions require at least $kn/2$ timeslots, since at each timeslot at most $2$ nodes transmit (due to \ex). Thus, in both time models, $\Omega(k)$ rounds are required. Moreover, in the synchronous time model, dissemination of a single message will take at least $D/2$ rounds, since in this model, a message can travel at most one hop in a single round. So, for the synchronous time model, the bound $\Theta(k+D)$ is tight and optimal.
\end{proof}

%Using the above result along with Claim \ref{clm:broadcast_is_at_least_l_max} gives us that for any constant maximum degree graphs, for any gossip protocol, dissemination of a single message will take $\Omega(D)$ rounds with high probability. Moreover, it requires $\Omega(k)$ rounds to disseminate $k$ messages, since at every timeslot, at most $2$ messages (in \ex) are sent and received (thus we need at least $kn$ timeslots). So, the bound for such graphs becomes: $\Theta(k+D)$ and hence it is tight and optimal.

\section{TAG: \texorpdfstring{$k$}{k}-dissemination with Tree-based Algebraic Gossip}\label{sec:tag}
%\subsection{Gossip dissemination protocol \texorpdfstring{TAG}{A} on arbitrary graphs and \textit{k} messages}

%\algrenewcommand{\algorithmiccomment}[1]{\hfill$\rightarrow$ #1}
\newcommand{\commentchar}{\ensuremath{/\mkern-4mu/}}
\algrenewcommand{\algorithmiccomment}[1]{\hskip3em $\commentchar$ #1}

\begin{algorithm}[t]
\renewcommand{\thealgorithm}{}
\floatname{algorithm}{Protocol TAG}
\begin{algorithmic}[1]
\Require $N(v)$, $k$, gossip spanning tree protocol $\S$
\renewcommand{\algorithmicrequire}{\textbf{Initialize:}}
\Require $parent=null$
\Statex

\Statex \ul{On odd wakeup:}\Comment Phase 1: \ex gossip spanning tree protocol $\S$
\State choose parter $u \in N(v)$ and exchange messages with it according to $\S$
\State according to $\S$ decide if $parent=u$
\Statex

\Statex \ul{On even wakeup:}\Comment Phase 2: \ex algebraic gossip
\If{obtained $parent$ during the protocol $\S$}
\State exchange messages with $parent$ according to algebraic gossip (RLNC)% with $comm.\text{ }partner=parent$
\EndIf
\Statex

\Statex \ul{On contact from other node $w\in N(v)$:}
\If{$w$ performs Phase 1}
%	\If{received non-empty broadcast message for the first time} $parent=u$ \EndIf
	\State exchange messages with $w$ according to $\S$
	\State according to $\S$ decide if $parent=w$
\Else ($w$ performs Phase 2)
%	\State store or discard received message according to algebraic gossip
	\State exchange messages with $w$ according to algebraic gossip (RLNC)
\EndIf
\end{algorithmic}
\caption{Pseudo code for node $v$. Example for asynchronous time model.}
%\label{algo:optimalAG}
\end{algorithm}

We now describe the protocol TAG (Tree based Algebraic Gossip),
%\hyperref[algo:optimalAG]{TAG},
which is a $k$-dissemination gossip protocol that exploits algebraic gossip in conjunction with a spanning tree gossip protocol $\S$ (see Sec. \ref{sec:pre}).
Given a connected network of $n$ nodes and $k$ messages $x_1,...,x_k$ that are initially located at some nodes, the goal of the protocol TAG is to disseminate all the $k$ messages to all the $n$ nodes.
The protocol consists of two phases. Both phases are performed simultaneously in the following way: if a node wakes up when the total number of its wakeups until now is even, it acts according to \emph{Phase 1} of the protocol. If the node wakes up when the total number of its wakeups until now is odd, it acts according to \emph{Phase 2} of the protocol.
\begin{itemize}
\item In Phase 1, a node performs a spanning tree gossip protocol $\S$. Once a node becomes a part of the spanning tree, it obtains a \textbf{parent}.
%where the message that should be disseminated, is the ID of a node that holds the initial message $x_1$. We assume that nodes are aware of the numbering of initial messages, and thus the node holding $x_1$, let it be a node $v$, knows to start disseminating its ID.
%Once a node receives, \emph{for the first time} a message with the ID of $v$, it marks the neighbor from which this message was received as its \textbf{parent}. In the synchronous time model, a node can receive the message \emph{for the first time} simultaneously from many neighbors. In that case it will choose one of them as a parent randomly.

\item In Phase 2, a node is idle until it obtains a \textbf{parent} in Phase 1. From now on, in Phase 2, the node will perform an \ex algebraic gossip protocol with a fixed communication partner -- its \textbf{parent}.
\end{itemize}

The following theorem gives an upper bound on the stopping time of the protocol TAG.

\begin{theorem}\label{thm:opt_gossip}
Let $t(\S)$ be the stopping time of the gossip spanning tree protocol $\S$ performed at Phase 1, and let $d(\S)$ be the diameter of the spanning tree created by $\S$.
For any connected graph $G_n$, the stopping time of the $k$-dissemination protocol TAG, is:
\begin{align}
&t(\text{TAG})=O(k+\log n+d(\S)+t(\S)) \text{ rounds}
%&t(\text{TAG})=O(k+\log n+t(\S)) \text{ rounds for the synchronous time model}
\end{align}
for synchronous and asynchronous time models, and \whp 
\end{theorem}

\begin{proof}
In order to prove this theorem, we will find the time needed to finish TAG, after Phase 1 is completed. Once Phase 1 is completed, every node knows its parent and thus, in Phase 2, we have the algebraic gossip \ex protocol on the spanning tree $T_n$, where communication partners of the nodes are their parents. The following lemma gives an upper bound on the stopping time of such a setting.

\begin{lemma}\label{lemma:algebraic_gossip_in_tree_single_node}
Let $T_n$ be a tree with $n$ nodes, rooted at the node $r$, with depth $l_{\max}$.
%created by performing a broadcast protocol $\S$ starting from the node $r$.
There are $k$ initial messages located at some nodes in the tree.
Consider algebraic gossip \ex protocol with the following communication model: the communication partner of a node is fixed to be its parent in $T_n$ during the whole protocol.
Then, the time needed for \textbf{all the nodes} to learn all the $k$ messages is $O(k+\log n+l_{\max}))$ rounds for the synchronous and asynchronous time models, with probability of at least $1-\tfrac{2}{n}$.
\end{lemma}

The proof of Lemma \ref{lemma:algebraic_gossip_in_tree_single_node} is very similar to the proof of Theorem \ref{thm:algebraic_gossip_with_k}, and relies on reducing the problem of algebraic gossip to a simple system of queues. The service time is geometrically distributed with a worst-case parameter $p=\tfrac{1}{2n}$. The $\Delta$ is eliminated from $p$ since each node chooses now a single communication partner. Then, using Theorem \ref{thm:tree_of_queues} we obtain the stopping time of algebraic gossip with on the tree $T_n$. Detailed proof of Lemma \ref{lemma:algebraic_gossip_in_tree_single_node} can be found in appendix.

Since for every choice of the tree root, the depth of the tree $T_n$ (which was created using protocol $t(\S)$) is bounded by its diameter, we can replace the $l_{\max}$ in the bound $O(k+\log n+l_{\max}))$ with $d(\S)$.
Now, we just add the stopping time of Phase 1 (the spanning tree time -- $t(\S)$) and the stopping time of Phase 2 (after Phase 1 has finished), and obtain that the number of rounds needed to complete the protocol TAG is $O(k+\log n+d(\S) +t(\S))$ \whp
%\begin{align}
%t(\text{TAG})=O(k+\log n+d(\S) +t(\S)) \text{ rounds}.
%\end{align}
\end{proof}

\subsection{TAG protocol using \texorpdfstring{$1$}{1}-dissemination as a spanning tree protocol}

The spanning tree task can be successfully performed by a simple gossip broadcast (or $1$-dissemination) protocol. When a node receives for the first time the message, it marks the sending node as its parent. In such a way we obtain a spanning tree rooted at the node that initiated the broadcast protocol. Let us denote a gossip $1$-dissemination protocol as $\mathcal{B}$. Clearly, the result of Theorem \ref{thm:opt_gossip} can be rewritten as: $t(\text{TAG})=O(k+\log n + d(\mathcal{B})+t(\mathcal{B}))$. An interesting observation regarding the broadcast protocol $\mathcal{B}$, is that for synchronous time model
the depth of the broadcast tree cannot be larger that the broadcast time (measured in rounds), i.e., $t(\mathcal{B})\ge d(\mathcal{B})$. The last is true since a message can not travel more than one hop in a single round. Thus, for the synchronous time model we obtain
that the number of rounds needed to complete the TAG protocol \whp is:
\begin{align}
t(\text{TAG})=
O(k+\log n + t(\mathcal{B})).
\end{align}

%\begin{corollary}
%\label{col:proto_A_sync}
%The upper bound for the protocol TAG: $O(k+\log n +d(\S)+ t(\S))$ (Theorem \ref{thm:opt_gossip}) holds also for the synchronous time model. Moreover, the bound can be reduced to: $O(k+\log n + t(\S))$ rounds.
%\end{corollary}
%
%\begin{proof}
%As in the proof of the Corollary \ref{col:uniform_ag_sync}, in the synchronous time model, the service time distribution parameter ($p$) will be larger by the $n$ factor, and the time will be measured in rounds instead of timeslots. Thus, using the same arguments as in the proof of Corollary \ref{col:uniform_ag_sync}, we obtain the upper bound of $O(k+\log n +d(\S)+ t(\S))$ rounds.
%
%Moreover, in the synchronous time model, a message can not travel more than one hop in a single round, thus, the depth of the broadcast tree cannot be larger that the broadcast time (measured in rounds), i.e., $t(\S)\ge d(\S)$. so, in the synchronous time model, the stopping time of the protocol TAG is: $O(k+\log n + t(\S))$ rounds.
%\end{proof}

\section{Optimal All-to-all Dissemination Using TAG}\label{sec:rr}
%\section{Linear gossip broadcast protocol \texorpdfstring{$\mathcal{B_{\mathcal{RR}}}$}{}}

In this section we propose to use the TAG protocol in conjunction with a $1$-dissemination (or broadcast) gossip protocol $\mathcal{B_{RR}}$ for spanning tree construction. For the case where $k=\Omega(n)$ messages need to be disseminated, TAG with $\mathcal{B_{RR}}$ achieves order optimal performance.
For the case $k=\Omega(n)$ the lower bound of any gossip dissemination protocol is $\Omega(n)$ rounds. The bound from Theorem \ref{thm:opt_gossip} gives $t(\text{TAG})=O(k+\log n + d(\S)+t(\S))$, and if $k=n$ we obtain $O(n +t(\S))$. Thus, all we need to show is the existence of a gossip spanning tree protocol that finishes after $O(n)$ rounds \whp on any graph.

\begin{theorem}\label{thm:linear_broadcast}
For any connected graph $G_n$, the stopping time of the broadcast protocol with the round-robin communication model -- $\mathcal{B_{RR}}$ is $O(n)$ rounds. In the asynchronous time model, this result holds with probability of at least $1-n(2/e)^{3n}$, and in the synchronous time model, with probability $1$.
\end{theorem}

In order to prove Theorem \ref{thm:linear_broadcast} we need the following lemma which is proved in the appendix.
\begin{lemma}\label{lemma:sum_of_degrees_on_shortest_path}
For any connected graph $G_n$ with $n$ nodes, the sum of the degrees of the nodes along any shortest path between any two nodes $v$ and $u$ is at most $3n$.
\end{lemma}

%The proof of the following Lemma can be found in Appendix.
%\begin{lemma}
%\label{lemma:geom_upper_bound}
%Let $X$ be a sum of $m$ independent and identically distributed geometric random variables (each one with parameter $p>0$) and $\text{E}\left[X\right]=\tfrac{m}{p}$.
%Then, for $\alpha>1$:
%\begin{align}
%\Pr \left(X \leq \alpha\text{E}\left[X\right]\right) > 1-\left(\alpha e^{1-\alpha}\right)^m.
%\end{align}
%\end{lemma}

%The proof of the following Lemma can be found in 'reference to our paper'.
%\begin{lemma}\label{lemma:sum_of_exp_bounded1}
%Let $Y$ be the sum of $m$ independent and identically distributed exponential random variables (each one with parameter $\mu>0$), and $\text{E}\left[Y\right]=\tfrac{m}{\mu}$.
%Then, for $\alpha>1$:
%$$\Pr \left(Y < \alpha\text{E}\left[Y\right]\right) > 1-(2e^{-\alpha/2})^m.$$
%\end{lemma}

\begin{proof}[Proof of Theorem \ref{thm:linear_broadcast}]
In this proof we assume the \push gossip variation, but it is clear that the result holds also for \ex.
Without loss of generality, assume that the message that needs to be disseminated is initially located at the node $v$.
In the \emph{round-robin} gossip, when a node is scheduled to transmit, it transmits a message to its neighbor according to the \emph{round robin} scheme. %I.e, at every transmission a message is sent to a different neighbor.

Consider a shortest path between $v$ and some other node $u$.
On the shortest path of length $l$ there is exactly one node at the distance $i$ from $v$, where $i\in\left[0,\dots,l\right]$, and $l\le n-1$.
Let $d_i$ be the degree of the node at distance $i$ from $v$.
In order to guarantee the delivery of the message from $v$ to $u$, we need $\sum_{i=0}^{l}d_i$ transmissions in the following order: first, we need $d_0$ transmissions of the node $v$, then $d_1$ transmissions of the next node in the path $v\rightarrow u$, and so on until the message is delivered to $u$. From Theorem \ref{lemma:sum_of_degrees_on_shortest_path}, $\sum_{i=0}^{l}d_i\le 3n$.

In the asynchronous model, a node transmits at a given \emph{timeslot} with probability $\tfrac{1}{n}$.
So, the number of timeslots until some specific node transmits is a geometric random variable with parameter $\tfrac{1}{n}$. We define this geometric random variable as $X$, i.e., $X\sim\text{Geom} \left(\tfrac{1}{n}\right)$.

The number of timeslots until $3n$ specific transmissions occur, is the sum of $3n$ independent geometric random variables. Using a Chernoff bound we obtain $O(n^2)$ timeslots (or $O(n)$ rounds) with exponential high probability. The last allows us to perform union bound for shortest paths to all other nodes in $G$, thus obtaining the $O(n)$ bound for the broadcast time. We omit here the formal part of the proof. The full proof can be found in the appendix.
%
%Now, let us write the above more formally:
%
%We define $t_{v\rightarrow u}$ as the time it takes to guarantee the delivery of the message from the node $v$ to an arbitrary node $u$. As we showed above, $t_{v\rightarrow u}$ is the number of timeslots until $3n$ specific transmissions occur, so:
%\begin{align}
%t_{v\rightarrow u}=\sum_{i=1}^{3n}X_i ,
%\end{align}
%\begin{align*}
%X_i \text{ } \forall i\in[1,\ldots,3n] \text{ are i.i.d. and distributed as }X,\\\text{where }X\sim\text{Geom} \left(\frac{1}{n}\right).
%\end{align*}
%Thus,
%\begin{align}
%\text{E}\left[t_{v\rightarrow u}\right]=\text{E}\left[\sum_{i=1}^{3n}X_i\right]=\sum_{i=1}^{3n}\text{E}\left[X_i\right]=\sum_{i=1}^{3n}n = 3n^2.
%\end{align}
%From Lemma \ref{lemma:geom_upper_bound} with $\alpha = 2$ :
%\begin{align}
%\Pr\left(t_{v\rightarrow u}\leq 2\text{E}\left[t_{v\rightarrow u}\right]\right)>1-(2/e)^{3n},
%\end{align}
%or
%\begin{align}
%\Pr\left(t_{v\rightarrow u}\leq 6n^2\right)>1-(2/e)^{3n}.
%\end{align}
%Now, we will apply a union bound on probabilities of the events: $t_{v\rightarrow u'}>6n^2$, where $u'\in{V}$. Notice, that $\text{E}\left[t_{v\rightarrow u}\right]=3n^2$ for all $u'\in{V}$.
%\begin{align}
%\Pr\left(\bigcup_{u'\in{V}}(t_{v\rightarrow u'}>6n^2)\right)\le \sum_{u'\in{V}}\Pr \left(t_{v\rightarrow u'}\right),
%\end{align}
%so,
%\begin{align}
%\Pr\left(\bigcup_{u'\in{V}}(t_{v\rightarrow u'}>6n^2)\right)\le n(2/e)^{3n}.
%\end{align}
%Thus,
%\begin{align}
%\Pr\left(\bigcap_{u'\in{V}}(t_{v\rightarrow u'}\le 6n^2)\right)>1-n(2/e)^{3n}.
%\end{align}
%So, we obtain the result of $O(n^2)$ timeslots, or $O(n)$ rounds.

It is easy to see that in the synchronous time model, $3n$ specific transmissions will occur exactly after $3n$ communication rounds. E.g., after $d_0$ rounds, $v$ will perform $d_0$ transmissions -- each one to different neighbor (according to the round-robin scheme). Thus, the message will be delivered to $u$ after at most $3n$ rounds with probability $1$.
\end{proof}

Using Theorems \ref{thm:opt_gossip} and \ref{thm:linear_broadcast} we obtain the upper bound on the stopping time of TAG with $\mathcal{B_{RR}}$ as a spanning tree construction protocol: $O(k+\log n + d(\S) + n)$ which is $\Theta(n)$ for $k=\Omega(n)$.

\section{Graphs with a Large Weak Conductance}\label{sec:weak}
For values of $k$ which are smaller than $n$ we use the information spreading protocol (hereafter, IS) of~\cite{censor2010fast}, which requires only a polylogarithmic number of rounds for broadcast on graphs with large \emph{weak conductance}. Roughly speaking, the weak conductance is a value in $[0,1]$ that measures the connectivity of subsets of nodes of a graph. It has been used to analyze the time required for \emph{partial} information spreading, where each message is only required to reach some fraction of the nodes. This, in turn, has been applied in the analysis of the IS protocol to show that the running time for full information spreading inversely depends on the weak conductance. The graphs with large weak conductance, for which the IS protocol is fast, form a broad family of graphs, including graphs that exhibit some (though not too many) communication bottlenecks. A simple example is the barbell graph, consisting of two cliques of $n/2$ nodes, connected by a single edge, which corresponds to a bottleneck since information must pass along it, but the probability of randomly choosing it is small due to large node degrees. The IS protocol overcomes this and runs in a logarithmic number of synchronous rounds on the barbell.

We describe this result for both the synchronous and asynchronous time models considered. Although the IS protocol is designed to disseminate $n$ messages originating one at each node, we will only use it for obtaining a spanning tree of our communication graph, while the actual information dissemination is done using algebraic gossip (i.e., we use the TAG protocol with IS as the spanning tree construction protocol). This is since the IS protocol sends large messages, while the goal of algebraic gossip is to address bandwidth concerns. The spanning tree is constructed as follows. The information sent by a node $v$ is an $n$-bit string, characterizing the nodes from which $v$ heard from, whether directly or indirectly. This corresponds to empty initial inputs, and initially the $n$-bit string of node $v$ is a unit vector, characterizing only the empty input of the node $v$ itself. The $n$-bit string maintained and sent by a node $v$ is monotone, in the sense that as time passes, its entries can only change from zero to one. The spanning tree that is created corresponds to each node $v$ declaring its parent as the first node $u$ from which it received a message that caused its most significant bit to change from zero to one. This means that this node received the input of the node $w$ corresponding to the most significant bit (recall that the input itself is an empty string).

The following theorem characterizes the time required for the IS protocol to complete.
\begin{theorem}[~\protect{\cite[Theorem 4.1]{censor2010fast}}]
For every $c>1$ and every $\delta \in (0,1/3c)$, the IS protocol obtains full information spreading after at most $O(c(\frac{\log{(n)}+\log{(\delta^{-1})}}{\Phi_c(G)}+c))$ rounds, with probability at least $1-3c\delta$.
\end{theorem}

In the synchronous model we can use the IS protocol in the TAG protocol, directly obtaining the following theorem, which shows optimality of TAG for certain families of parameters.
%\begin{theorem}
%\label{theorem:sync-info-spr}
%Let $c=O(\mbox{polylog}(n))$, let $G$ be a graph with weak conductance $\Phi_c=\Omega(\frac{1}{\mbox{polylog}(n)})$, and let $k=O(\mbox{polylog}(n))$. With probability at least $1-\frac{1}{n}$, the time for disseminating $k$ messages using algorithm $\cal{A}$ in conjunction with the IS algorithm is $O(\mbox{polylog}(n))$ synchronous rounds.
%\end{theorem}
%\chen{we we actually need is the optimality case:}
\begin{theorem}
\label{theorem:sync-info-spr}
Let $c=O(\log^p{(n)})$ for some $p\geq 0$, let $G$ be a graph with weak conductance $\Phi_c=\Omega(\frac{1}{\log^p{(n)}})$, and let $k=\Omega(\log^{2p+1}{(n)})$. With probability at least $1-\frac{1}{n}$, the time for disseminating $k$ messages using protocol TAG in conjunction with the IS protocol is $\Theta(k)$ synchronous rounds.
\end{theorem}

We show that the IS protocol works in the asynchronous model as well. While this is not a direct usage of the protocol due to some subtleties, we nevertheless show how to obtain our result as for the synchronous model. Our analysis induces an overhead of $O(\log^2(n))$ rounds.

We do not change the protocol itself to cope with asynchrony, but rather analyze the time required using additional techniques. Roughly speaking, the outline of our analysis is showing that segments of the asynchronous execution simulate synchronous rounds. This allows us to use the original analysis of the protocol for the simulated rounds, which gives our result, as stated in the following theorem, and proved in the appendix.

%\chen{This is not true since we dont know that the diameter of the broadcast graph is $O(\mbox{polylog}(n))$.}
\begin{theorem}
\label{theorem:async-info-spr}
Let $c=O(\log^p{(n)})$ for some $p\geq 0$, let $G$ be a graph with weak conductance $\Phi_c=\Omega(\frac{1}{\log^p{(n)}})$, and let $k=\Omega(\log^{2p+3}{(n)})$. With probability at least $1-\frac{1}{n}$, the time for disseminating $k$ messages using protocol TAG in conjunction with the IS protocol is $O(k+l_{\max})$ rounds for the asynchronous time model, where $l_{\max}$ is the depth of the spanning tree induced by the IS protocol.
\end{theorem}

For completeness, we note that, in IS, during the even-numbered steps of each node the choice of neighbor to contact is randomized. For these steps alone, adapting the analysis Mosk-Aoyama and Shah~\cite{Mosk-Aoyama2006Computing} for the asynchronous case to our protocol, implies that the extra $\log{(n)}$ time slots can be avoided for the purpose of partial information spreading alone (as used in the proof of the information spreading protocol (see~\protect{\cite[Theorem 2.2]{censor2010fast}}). However, as this cost is required anyhow to argue about the deterministic choices, made during the odd-numbered steps, we omit going through this adjustment.

\bibliographystyle{abbrv}
\bibliography{gossip_opt}

\appendix
\label{appendix}
\section*{Appendix}
\addcontentsline{toc}{section}{Appendix}
\subsection*{Table of notations}
\renewcommand{\arraystretch}{1.3}
\begin{table}[thbp]
	\centering
	\scalebox{0.9}{
		\begin{tabular}{|c||l|}
		\hline
			$n$ & Number of nodes\\\hline
			$k$ & Number of messages needed to be disseminated\\\hline
			$G_n$ & Connected graph with $n$ nodes\\
			\hline
			$T_n$ & Connected Tree graph with $n$ nodes\\
			\hline
			$Q_n^{tree}$ & Network of $n$ queues arranged in a tree topology\\\hline
			$Q_{l_{\max}}^{line}$ & Network of $l_{\max}$ queues arranged in a line topology\\\hline
			$D$ & Diameter of a graph\\\hline
			$N(v)$ & Set of neighbors of the node $v$\\\hline
			$d_v$ & Degree of the node $v$ ($d_v=|N(v)|$)\\\hline
			$\Delta$ & Maximum degree of the graph ($\Delta=\max_v d_v$)\\\hline
			\textbf{timeslot} & Unit of time in the \textbf{asynchronous} time model\\\hline
			\textbf{round} & Unit of time in the \textbf{synchronous} time model ($1$ round = $n$ timeslots)\\\hline
			$\S$ & Some spanning tree gossip protocol\\\hline
			$\mathcal{B}$ & Some broadcast ($1$-dissemination) gossip protocol\\\hline
			$d(\S)$, $d(\mathcal{B})$ & Diameter of the spanning tree created by the protocol\\\hline
			%$\mathcal{B^*}$ & Optimal broadcast gossip protocol\\\hline
			$\mathcal{RR}$ & Round-robin communication model\\\hline
			$\mathcal{B_{RR}}$ & Broadcast gossip algorithm based on the round-robin communication model\\\hline
			$l_{\max}$ & Depth of the tree created by a broadcast protocol\\\hline
			TAG & $k$-dissemination protocol that uses algebraic gossip and a spanning tree protocol\\\hline
			
			$t(\text{TAG})$, $t(\S)$, $t(\mathcal{B})$ & Stopping time of a protocol\\\hline
			$t(Q_n^{tree})$ & Stopping time of a queuing system -- time by which all customers leave the system\\\hline
		\end{tabular}
		}
	\caption{Table of notations}
	\label{tab:Notations}
\end{table}

%\begin{figure*}
%\begin{center}
%\input{fig/paper_flowchart.tex}
%\end{center}
%\caption{Flowchart of the paper.}
%\label{fig:flowchart}
%\end{figure*}

\subsection*{Proof of Lemma \ref{lemma:algebraic_gossip_in_tree_single_node}}
\addcontentsline{toc}{subsection}{Proof of Lemma \ref{lemma:algebraic_gossip_in_tree_single_node}}

\begin{lemma:algebraic_gossip_in_tree_single_node}[restated]
\label{lemma:algebraic_gossip_in_tree_single_node_restated}
Let $T_n$ be a tree with $n$ nodes, rooted at the node $r$, with depth $l_{\max}$.
%created by performing a broadcast protocol $\S$ starting from the node $r$.
There are $k$ initial messages located at some nodes in the tree.
Consider algebraic gossip \ex protocol with the following communication model: the communication partner of a node is fixed to be its parent in $T_n$ during the whole protocol.
Then, the time needed for \textbf{all the nodes} to learn all the $k$ messages is $O(k+\log n+l_{\max}))$ rounds for the synchronous and asynchronous time models, with probability of at least $1-\tfrac{2}{n}$.
\end{lemma:algebraic_gossip_in_tree_single_node}

\begin{proof}
The proof is very similar to the proof of Theorem \ref{thm:algebraic_gossip_with_k}, and relies on reducing the problem of algebraic gossip to a simple system of queues.

%Let $l_{\max}$ be the depth of the tree $T_n$ rooted at $r$.
On $T_n$, consider a message flow towards an arbitrary node $v$ (\emph{not necessary the root} of $T_n$) from all other nodes. Once $k$ \emph{helpful messages} arrive at $v$, it will reach the rank $k$ and finish the algebraic gossip protocol.
Due to the proposed communication model, every node in $T_n$ has a fixed communication partner -- its parent, so, each edge $e$ in the tree has at least one node which will issue, on its wakeup, a bidirectional communication (\ex) over $e$. Thus, from every node, a message can be sent towards $v$. We ignore messages that are not sent in the direction of $v$. Ignoring part of messages can only increase the stopping time of the algebraic gossip protocol.

As in the proof of Theorem \ref{thm:algebraic_gossip_with_k}, we define a queuing system $Q_n^{tree}$ by assuming an infinite queue with a single server at each node. The root of $Q_n^{tree}$ will be an arbitrary node $v$, and let $l_{\max}^v$ be the depth of the tree $Q_n^{tree}$.

%Customers of our queuing network are \emph{helpful messages}, i.e., messages that increase the rank of a node they arrive at. This means that every customer arriving at some node increases its rank by 1. The queue length of a node represents a measure of \emph{helpfulness} of the node to its neighbor, i.e., the number of \emph{helpful messages} it can generate for it.

The service procedure at a node is a transmission of a \emph{helpful message} towards the node $v$. In our communication model, the communication partner of a node \emph{is always its parent} in the tree. Thus, in the \ex gossip variation, in the asynchronous time model, in a given timeslot, a \emph{helpful message} will be sent over the edge in a specific direction with probability of at least $(1-\tfrac{1}{q})/n$, where $\tfrac{1}{n}$ is the probability that a given node wakes up in a given timeslot, and $1-\tfrac{1}{q}$ is the minimal probability that the message will be \emph{helpful}.
Thus, we can consider that the service time in our queuing system is geometrically distributed with parameter $p\ge (1-\tfrac{1}{q})/n$, and since $q\ge 2$, we can assume the worst case: $p=\tfrac{1}{2n}$.

%Lemma 2 in \cite{Borokhovich2010Tight} shows that we can model the service time of each server as an exponential random variable with parameter $\mu=p$.

Using Theorem \ref{thm:tree_of_queues} for the tree $T_n$ rooted at $v$, with $\mu=p$, we get an upper bound for the stopping time of the node $v$, $t_v=O((k+l_{\max}^v+\log n)2n)$ timeslots with probability of at least $1-\tfrac{2}{n^2}$, where the $l_{\max}^v$ is the depth of the tree $T_n$ rooted at $v$. Since $l_{\max}^v\le 2l_{\max}$ (where $l_{\max}$ is the depth of $T_n$ rooted at $r$), we can replace the $l_{\max}^v$ with $2l_{\max}$.

So, using union bound, we obtain the upper bound (measured in timeslots) for all the nodes in $T_n$:

\begin{align}
\Pr\left(\bigcap_{v\in V} t_v=O((k+\log n + l_{\max})2n)\right)>1-\frac{2}{n}.
\end{align}

As in the proof of Theorem \ref{thm:algebraic_gossip_with_k}, in the synchronous time model, the service time distribution parameter $p$ will be larger by a factor of $n$, and the time will be measured in rounds instead of timeslots. Thus, using the same arguments as in the proof of Theorem \ref{thm:algebraic_gossip_with_k}, we obtain the upper bound of $O(k+\log n + l_{\max})$ rounds for the synchronous time model. Thus, the lemma holds for both time models.
\end{proof}

\subsection*{Proof of Lemma \ref{lemma:sum_of_degrees_on_shortest_path}}
\addcontentsline{toc}{subsection}{Proof of Lemma \ref{lemma:sum_of_degrees_on_shortest_path}}

\begin{lemma:sum_of_degrees_on_shortest_path}[restated]
\label{lemma:sum_of_degrees_on_shortest_path_restated}
For any connected graph $G_n$ with $n$ nodes, the sum of the degrees of the nodes along any shortest path between any two nodes $v$ and $u$ is at most $3n$.
\end{lemma:sum_of_degrees_on_shortest_path}

\begin{proof}
Without loss of generality, consider a BFS spanning tree of $G$ rooted at some node $v$, and some arbitrary leaf $u$.
We will find the maximum degree of the node located on the path $(v\rightarrow u)$ at distance $i$ from the root $v$. Clearly, such a node can be connected only to the following nodes:
\begin{itemize}
\item Nodes that are located at distance $i-1$ from the root. (It can not be connected to the nodes that are closer to the root (than $i-1$) since then, its distance from the root would be $i-1$ which contradicts the given BFS execution.)
\item Nodes that are at the same distance $i$ from the root.
\item Nodes that are located at distance $i+1$ from the root. (It can not be connected to the nodes that are farther from the root (than $i+1$) since then, their distance from the root would be $i+1$ which contradicts the given BFS execution.)
\end{itemize}
Let us define $m_i$ as the number of nodes at distance $i$ from the root. Clearly, $\sum_{i=0}^{n-1}m_i=n$. (The node at distance $0$ is the root $v$).
The degree of a node (at distance $i$ from the root) can be at most: $d_i\le(m_{i-1}+m_{i}+m_{i+1})$. Thus, the sum of degrees on a path of length $l$ from the root to a leaf is at most: $d=\sum_{i=0}^{l}d_i$. Since $l\le n-1$, $d=\sum_{i=0}^{l}d_i\le \sum_{i=0}^{n-1}d_i = \sum_{i=0}^{n-1}(m_{i-1}+m_{i}+m_{i+1}) \le 3n$.
\end{proof}

\subsection*{Proof of Claim \ref{clm:D_in_const_delta}}
\addcontentsline{toc}{subsection}{Proof of Claim \ref{clm:D_in_const_delta}}

\begin{clm:D_in_const_delta}[restated]
\label{clm:D_in_const_delta_restated}
For any connected graph $G_n$ with a constant maximum degree $(\Delta=O(1))$, the diameter of $G_n$ is $\Omega(\log n)$.
\end{clm:D_in_const_delta}

\begin{proof}
Let us sum up all the $n$ vertices of $G_n$ in the following way. We start with an arbitrary node $v$ and count it as 1. Then we split the sum of $n$ vertices into $D$ parts, where $D$ is the diameter of $G_n$. Each part represents number of vertices located at the distance $i$ ($i\in[0,..,D]$) from the node $v$. Since we are interested in the lower bound on $D$, we can assume the maximum degree for every node (so, the number of parts in the sum will be minimal).
We define $n_i$ ($i\in[0,..,D]$) as the number of vertices located at the distance $i$ from the node $v$.
Thus we obtain:
\begin{align}
n_0+n_1+n_2+\dots+n_D &= n \\
1+\Delta+\Delta^2+\dots+\Delta^D &\ge n \\
\frac{\Delta^{D+2}-1}{\Delta-1} &\ge n \\
\Delta^{D+2} &\ge n \\
D+2 &\ge \log_{\Delta}n \\
D &=\Omega(\log n)
\end{align}
\end{proof}

\subsection*{Stochastic Dominance}
\addcontentsline{toc}{subsection}{Stochastic Dominance}

%First, we will present some notations and claims used in the proof.

\begin{dfn}[Stochastic dominance, stochastic ordering \cite{Hofstad08randomgraphs,GrimmettStirzaker:01}]
\label{dfn:stoch_dom}
We say that a random variable $X$ is stochastically less than or equal to a random variable $Y$ if and only if $\Pr(X\le t)\ge \Pr(Y\le t)$, and such a relation is denoted as: $X\preceq Y$.
\end{dfn}

\begin{clm}
\label{clm:stoch_dom_max}
If for $i\in\{1,2\}$, $X_i\preceq Y_i$, $X_i$ are independent, and $Y_i$ are independent, then: $\max_i{X_i}\preceq \max_i{Y_i}$.
\end{clm}

\begin{proof}
\begin{align*}
\Pr(\max_i{X_i}\le t)=\bigcap_i\Pr(X_i\le t)= \prod_i\Pr(X_i\le t)\\\ge \prod_i\Pr(Y_i\le t)=\Pr(\max_i{Y_i}\le t).
\end{align*}
Hence:
\begin{align*}
\max_i{X_i}\preceq \max_i{Y_i}.
\end{align*}
\end{proof}

\begin{clm}
\label{clm:stoch_dom_sum}
If for $i\in\{1,2\}$, $X_i\preceq Y_i$, $X_i$ are independent, and $Y_i$ are independent, then: $\sum_i{X_i}\preceq \sum_i{Y_i}$.
\end{clm}

\begin{proof}
\begin{align*}
\Pr(X_1+X_2\le t)=\int_{-\infty}^t f_{X_1+X_2}(s)ds,\\\text{where } f_{X_1+X_2}(s)=f_{X_1}(s)\ast f_{X_2}(s).
\end{align*}
Thus:
\begin{align}
\Pr(X_1+X_2\le t)&=\int_{-\infty}^t \int_{-\infty}^{\infty}f_{X_1}(\tau)f_{X_2}(s-\tau)d\tau ds
\\&=\int_{-\infty}^{\infty}f_{X_1}(\tau)\Pr(X_2\le t-\tau) d\tau
\\&\ge\int_{-\infty}^{\infty}f_{X_1}(\tau)\Pr(Y_2\le t-\tau) d\tau
\\&=\int_{-\infty}^t \int_{-\infty}^{\infty}f_{X_1}(\tau)f_{Y_2}(s-\tau)d\tau ds
\\&=\int_{-\infty}^t \int_{-\infty}^{\infty}f_{Y_2}(\tau)f_{X_1}(s-\tau)d\tau ds
\\&=\int_{-\infty}^{\infty}f_{Y_2}(\tau)\Pr(X_1\le t-\tau) d\tau
\\&\ge\int_{-\infty}^{\infty}f_{Y_2}(\tau)\Pr(Y_1\le t-\tau) d\tau
\\&=\int_{-\infty}^t \int_{-\infty}^{\infty}f_{Y_2}(\tau)f_{Y_1}(s-\tau)d\tau ds
\\&=\Pr(Y_1+Y_2\le t).
\end{align}
Hence:
\begin{align*}
\sum_{i=1}^2{X_i}\preceq \sum_{i=1}^2{Y_i}.
\end{align*}
\end{proof}

\subsection*{Later arrivals yield later departures}
\addcontentsline{toc}{subsection}{Later arrivals yield later departures}

Consider an infinite FCFS queue with a single exponential server. We define $a_i$ as the time of arrival number $i$ to the queue, and $d_i$ as time of the departure number $i$ from the queue. Let $X_i$ be the exponential random variable representing the service time of the arrival $i$. For all $i$, $X_i$'s are $i.i.d$.

Let $a_i$ be a sequence of $m$ arrival times to the queue, and $d_i$ be a sequence of $m$ departure times from the queue.

\begin{figure}[h]
\centering
\begin{tikzpicture}
[inner sep=0.6mm, place/.style={circle,draw=black,fill=blue!20,thick,minimum size=1cm},>=stealth]

\shade[left color=yellow!20,right color=blue!20] (0,0) rectangle (2,1);
\draw[black] (0,0) -- (2,0);
\draw[black] (0,1) -- (2,1);
\draw[black] (2,0) -- (2,1);
\node at (2.5,0.5) [place] (server)  {$\mu$};
\node at (0,0.5) [auto] (tail)  {\tiny{}};
\node at (4,0.5) [auto] (end)  {\tiny{}};
\node at (-1,0.5) [auto] (start)  {\tiny{}};

\path[->,thick] (start) edge node [above]{\small{$a_i$}} (tail);
\path[->,thick] (server) edge node [above]{\small{$d_i$}} (end);

\draw[black,fill=blue!10] (-1,-2) rectangle node[black] {\small{$X_1$}}(-0.3,-2.7) ;
\draw[black,fill=blue!10] (-0.3,-2) rectangle node[black] {\small{$X_2$}}(0.9,-2.7);
\draw[black,fill=blue!10] (1.2,-2) rectangle node[black] {\small{$X_3$}}(1.9,-2.7);
\draw[black,fill=blue!10] (2.4,-2) rectangle node[black] {\small{$X_4$}}(3.3,-2.7);
\draw[black,fill=blue!10] (3.3,-2) rectangle node[black] {\small{$X_5$}}(4,-2.7);

\draw[black,->] (-1,-3) -- (4.4,-3)node[black,right] {$t$};
\draw[black] (-1,-3.1) -- (-1,-2.9);

\draw[black,->,thick] (-0.3,-3.5) node[black,below] {$d_1$}-- (-0.3,-2.7);
\draw[black,->,thick] (0.9,-3.5) node[black,below] {$d_2$}-- (0.9,-2.7);
\draw[black,->,thick] (1.9,-3.5) node[black,below] {$d_3$}-- (1.9,-2.7);
\draw[black,->,thick] (3.3,-3.5) node[black,below] {$d_4$}-- (3.3,-2.7);
\draw[black,->,thick] (4,-3.5) node[black,below] {$d_5$}-- (4,-2.7);

\draw[black,->,thick] (-1,-1.2) node[black,above] {$a_1$}-- (-1,-2);
\draw[black,->,thick] (-0.3,-1.2) node[black,above] {$a_2$}-- (-0.3,-2);
\draw[black,->,thick] (1.2,-1.2) node[black,above] {$a_3$}-- (1.2,-2);
\draw[black,->,thick] (2.4,-1.2) node[black,above] {$a_4$}-- (2.4,-2);
\draw[black,->,thick] (3.3,-1.2) node[black,above] {$a_5$}-- (3.3,-2);

\node at (1.5,-4.8) [auto] (eq)  {$d_i=\max(a_i,d_{i-1})+X_i$};

\end{tikzpicture}
\caption{Arrival and departure times.}
\label{fig:arrival_and_departure}
\end{figure}
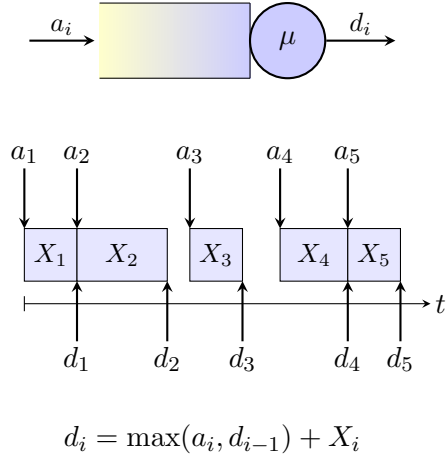

\begin{lemma}\label{lemma:later_arrivals}
If the sequence $a_i$ is replaced with another sequence of $m$ arrivals -- $\hat{a_i}$, such that: $\hat{a}_i\succeq a_i$ $\forall i\in[1,...,m]$, then, the resulting sequence of $m$ departures will be such that: $\hat{d}_i\succeq  d_i$ $\forall i\in[1,...,m]$. I.e., if every new arrival occurred ,stochastically, at the same time or later than the old arrival, then, every new departure from the queue will occur ,stochastically, at the same time or later than the old departure.
\end{lemma}

\begin{proof}
The proof is by induction on the arrival index $j$, $j\in[1,...,m]$.
\begin{itemize}
\item Induction basis: $\hat{d}_1\succeq d_1$. Follows since $d_1=a_1+X_1$, $\hat{d}_1=\hat{a}_1+X_1$, and $\hat{a}_1\succeq a_1$.
\item Induction assumption: $\forall i<j$ : $\hat{d}_i\succeq d_i$.
\item Induction step: we need to show that $\hat{d}_j\succeq d_j$.
\end{itemize}
If the $j$'s arrival occurred when the server was busy, then $d_j=d_{j-1}+X_j$. If the server was idle when the $j$'s arrival occurred, then $d_j=a_j+X_j$. Thus, we can write:
\begin{align}
d_j=\max(d_{j-1},a_j)+X_j, \\\text{and } \hat{d}_j=\max(\hat{d}_{j-1},\hat{a}_j)+X_j.
\end{align}
Since from induction assumption: $\hat{d}_{j-1}\succeq d_{j-1}$, and $\hat{a}_j\succeq a_j$, using Claims \ref{clm:stoch_dom_max} and \ref{clm:stoch_dom_sum}, we obtain $\hat{d}_j\succeq d_j$.
\end{proof}

\subsection*{Proof of Theorem \ref{thm:tree_of_queues}}
\addcontentsline{toc}{subsection}{Proof of Theorem \ref{thm:tree_of_queues}}

\begin{thm:tree_of_queues}[restated]
\label{thm:tree_of_queues_restated}
Let $Q_n^{tree}$ be a network of $n$ nodes arranged in a tree topology, rooted at the node $v$. The depth of the tree is $l_{\max}$. Each node has an infinite queue, and a single exponential server with parameter $\mu$. The total amount of customers in the system is $k$ and they are initially distributed arbitrarily in the network. The time by which all the customers leave the network via the root node $v$ is $t({Q}_n^{tree})=O((k+l_{\max}+\log n)/\mu)$ timeslots with probability of at least $1-\tfrac{2}{n^2}$.
\end{thm:tree_of_queues}

\begin{proof}

We denote the nodes of the queuing system $Q_n^{tree}$ as $Z_j^l$, where $l$ ($l\in[1,...,l_{\max}]$) is the level of the node in the tree, and $j$ is the node's index in the level $l$. The root of the $Q_n^{tree}$ tree is the node $Z_1^1$. All servers in the $Q_n^{tree}$ network are ON all the time (work-conserving scheduling), i.e., servers work whenever they have customers to serve. There are no external arrivals to the system. Once a customer is serviced on the level $l$, it enters the appropriate queue at the level $l-1$. When a customer is serviced by the root $Z_1^1$, it leaves the network.

Now, let us define the auxiliary queuing systems: $\hat{Q}_n^{tree}$ and $Q_{l_{\max}}^{line}$.

\begin{dfn}[Network $\hat{Q}_n^{tree}$]
\label{dfn:network_tn_hat}
$\hat{Q}_n^{tree}$ is the same network as $Q_n^{tree}$ with the following change in the servers' scheduling:

At any given moment, only one server at every level $l$ ($l\in[1,...,l_{\max}]$) is ON. Once a customer leaves level $l$,
% for one customer service period.
a server that will be scheduled (turned ON) at the level $l$, is the server which has in its queue a customer that has earliest arrival time to a queue at the level $l$ among all the current customers at the level $l$. If there are customers that initially reside at the level $l$, they will be serviced by the order of their IDs (we assume for analysis that every customer has a unique identification number).
\end{dfn}

\begin{dfn}[Network of queues $Q_{l_{\max}}^{line}$]
\label{dfn:line_of_queues}
$Q_{l_{\max}}^{line}$ is the the following modification of the network $Q_n^{tree}$, that results in a network of $l_{\max}$ queues arranged in a line topology.

For all $l\in[1,..,l_{\max}]$, we merge all the nodes at the level $l$ to a single node (a single queue with a single server). We name this single node at the level $l$ as the first node in $Q_n^{tree}$ at the level $l$, i.e., $Z_1^l$.
The customers that initially reside at level $l$ will be placed in a single queue in the order of their IDs.
This modification results in $Q_{l_{\max}}^{line}$ -- a network of $l_{\max}$ queues arranged in a line topology:
$Z_1^{l_{\max}}\rightarrow Z_1^{l_{\max}-1}\rightarrow\cdots\rightarrow Z_1^1$.
\end{dfn}

\begin{dfn}[Network of queues $\grave{Q}_{l_{\max}}^{line}$]
\label{dfn:line_of_queues_one_customer_back}
$\grave{Q}_{l_{\max}}^{line}$ -- is the same system as $Q_{l_{\max}}^{line}$ with the following modification.
We take the last customer at some node $Z_1^m$ ($m\in{[1,..,l_{\max}-1}]$) and place it at the head of the queue of the node $Z_1^{m+1}$. I.e., we move one customer, one queue backward in the line of queues.
\end{dfn}

\begin{dfn}[Network of queues $\hat{Q}_{l_{\max}}^{line}$]
\label{dfn:line_of_queues_all_customer_back}
$\hat{Q}_{l_{\max}}^{line}$ -- is the same system as $Q_{l_{\max}}^{line}$ with the following modification.
We move all the customers to the queue $Z^{l_{\max}}_1$. I.e., all the customers have to traverse now through all the $l_{\max}$ queues in the line.
\end{dfn}

We summarize the queuing systems defined above in the short Table \ref{tab:queuing_systems}.

\renewcommand{\arraystretch}{1.2}
\begin{table}[h]
	\centering
	\scalebox{0.9}{
		\begin{tabular}{|c||m{11.5cm}|}
		\hline
			${Q}_n^{tree}$ & Original system of $n$ queues arranged in a tree topology. Fig. \ref{fig:three_network_systems} (a).\\\hline
			$\hat{Q}_n^{tree}$ & System of $n$ queues arranged in a tree topology. Only one server is active at each level at a given time. Fig. \ref{fig:three_network_systems} (b).\\\hline
			${Q}_{l_{\max}}^{line}$ & System of $l_{\max}$ queues arranged in a line topology. Fig. \ref{fig:three_network_systems} (c).\\\hline
			$\grave{Q}_{l_{\max}}^{line}$ & System of $l_{\max}$ queues arranged in a line topology. One customer is moved one queue backward.\\\hline
			$\hat{Q}_{l_{\max}}^{line}$ & System of $l_{\max}$ queues arranged in a line topology. All customers are moved backward to the queue $Z^{l_{\max}}_1$.\\\hline
		\end{tabular}
		}
	\caption{Queuing systems used in the proof.}
	\label{tab:queuing_systems}
\end{table}

The proof of Theorem \ref{thm:tree_of_queues} consists of showing the following relations between the stopping times of the queuing systems:
%\begin{itemize}
%\item Show that the stopping time of the queuing system $Q_n^{tree}$, is at least as the stopping time of the queuing system $\hat{Q}_n^{tree}$, i.e., $t(Q_n^{tree})\preceq t(\hat{Q}_n^{tree})$.
%\item Show that the stopping time of the queuing system $\hat{Q}_n^{tree}$, is the same as the stopping time of the queuing system $Q_{l_{\max}}^{line}$, i.e., $t(\hat{Q}_n^{tree})\approx t({Q}_{l_{\max}}^{line})$.
%\item Show that the stopping time of the queuing system $Q_{l_{\max}}^{line}$ cannot become smaller if we move all the customers to the last queue in the line -- as in the system $\hat{Q}_{l_{\max}}^{line}$, i.e., $t(Q_{l_{\max}}^{line})\preceq t(\hat{Q}_{l_{\max}}^{line})$.
%\item Use Jackson's theorem for open networks to bound the time it takes to customers to cross the line of queues $\hat{Q}_{l_{\max}}^{line}$.
%\end{itemize}
%
%To summarize the above, we will show that:
$$t(Q_n^{tree})\preceq t(\hat{Q}_n^{tree})\approx t({Q}_{l_{\max}}^{line})\preceq t(\grave{Q}_{l_{\max}}^{line})\preceq t(\hat{Q}_{l_{\max}}^{line})=O((k+\log n + l_{\max})/\mu).$$

Stopping time of a queuing system $t(Q)$, is the time by which the last customer leaves the system (via the node $Z_1^1$).
In order to compare the stopping times of queuing systems, we define the following ordered set (or sequence) of departure time from a server $Z$ in a queuing system $Q$:
$d(Z,Q)=(d_1(Z,Q),d_2(Z,Q),...,d_i(Z,Q),...)$, where $d_i(Z,Q)$ is the time of the departure number $i$ from the node (server) $Z$.

%[\textbf{Michael}: where the set of arrival times to a node $Z_1^l$ is $\bigcup_{j}\hat{d}_i(Z_j^{l+1})$.]

\begin{figure}[h]
\centering
\input{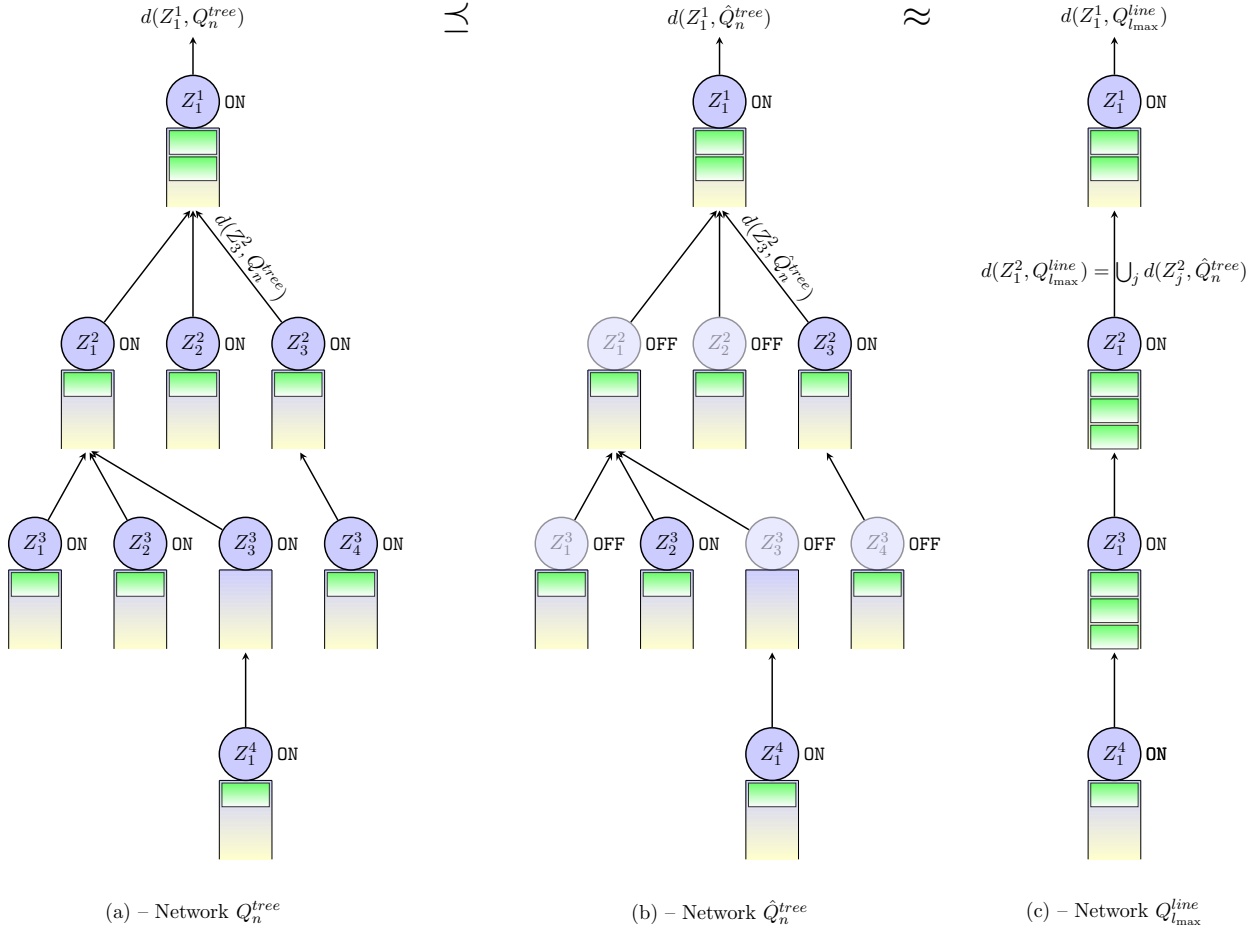}
\caption{(a) -- Network ${Q}_n^{tree}$, where all the servers work all the time.
(b) -- Network $\hat{Q}_n^{tree}$, where only one server at each level works at a given time.
(c) -- Network ${Q}_{l_{\max}}^{line}$.}
\label{fig:three_network_systems}
\end{figure}

First, we want to show that the stopping time of $Q_n^{tree}$ is at most the stopping time of the system $\hat{Q}_n^{tree}$, i.e., $t(Q_n^{tree})\preceq t(\hat{Q}_n^{tree})$.

\begin{lemma}\label{lemma:idle_server}
%Let $d(Z_j^l)$ and $\hat{d}(Z_j^l)$ be a sequences of departure times from the node $Z_j^l$ in ${Q}_n^{tree}$ and in $\hat{Q}_n^{tree}$ respectively.
In $\hat{Q}_n^{tree}$, every departure from the system (via $Z_1^1$) will occur, stochastically, at the same time or later than in ${Q}_n^{tree}$:
\begin{align}
{d}_i(Z_1^1,\hat{Q}_n^{tree})\succeq d_i(Z_1^1,{Q}_n^{tree}) \text{ } \forall i\in[1,...,k].
\end{align}
Thus, in $\hat{Q}_n^{tree}$, the last customer will leave the system, stochastically, at the same time or later than in ${Q}_n^{tree}$, or: $t({Q}_n^{tree})\preceq t(\hat{Q}_n^{tree})$.
\end{lemma}

\begin{proof}
The proof is by induction on the tree level $l$, $l\in[1,...,l_{\max}]$.
\begin{itemize}
\item Induction basis: $\forall i,j  \text{ : } {d}_i(Z_j^{l_{\max}},\hat{Q}_n^{tree})\succeq d_i(Z_j^{l_{\max}},{Q}_n^{tree})$. This is true since in $\hat{Q}_n^{tree}$, the nodes do not work all the time, and thus the departures will occur, stochastically, at the same time or later than in ${Q}_n^{tree}$. If there is a single node at the level $l_{\max}$, in $\hat{Q}_n^{tree}$ it will be ON all the time as in ${Q}_n^{tree}$, and thus, the departures will occur, stochastically, at the same time in both systems.
\item Induction assumption: for all $l>m$ ($m\ge 1$), $\forall i,j  \text{ : } {d}_i(Z_j^{l},\hat{Q}_n^{tree})\succeq d_i(Z_j^{l},{Q}_n^{tree})$.
\item Induction step: we need to show that: $\forall i,j  \text{ : } {d}_i(Z_j^{m},\hat{Q}_n^{tree})\succeq d_i(Z_j^{m},{Q}_n^{tree})$.
\end{itemize}
By induction assumption, for $l=m+1$: $\forall i,j  \text{ : } {d}_i(Z_j^{m+1},\hat{Q}_n^{tree})\succeq d_i(Z_j^{m+1},{Q}_n^{tree})$.
Now let us take a look at the departures from a node $Z_j^{m}$. There are two cases: $Z_j^{m}$ is a leaf, and $Z_j^{m}$ is not a leaf.
If $Z_j^{m}$ is a leaf, we can use the same argument as in the induction basis: in $\hat{Q}_n^{tree}$, the node $Z_j^{m}$ does not work all the time, and thus the departures from it in $\hat{Q}_n^{tree}$ cannot occur earlier than in ${Q}_n^{tree}$.
If $Z_j^{m}$ is not a leaf, it has input/inputs of arrivals from the level $m+1$.
Since the arrivals from the level $m+1$ in $\hat{Q}_n^{tree}$ occur, stochastically, at the same time or later than in ${Q}_n^{tree}$ (by induction assumption), even if the node $Z_j^{m}$ would work all the time (as in ${Q}_n^{tree}$), we would obtain from Lemma \ref{lemma:later_arrivals}: $\forall i,j  \text{ : } {d}_i(Z_j^m,\hat{Q}_n^{tree})\succeq d_i(Z_j^m,{Q}_n^{tree})$. Moreover, in $\hat{Q}_n^{tree}$, the node $Z_j^m$ does not work all the time (unless it is the only node at the level $m$), thus the departure times in $\hat{Q}_n^{tree}$ can be even larger.
\end{proof}

\begin{lemma}\label{lemma:tree_as_line}
In ${Q}_{l_{\max}}^{line}$, every departure from the system (via $Z_1^1$) will occur, stochastically, at the same time as in $\hat{Q}_n^{tree}$.
Thus, in ${Q}_{l_{\max}}^{line}$, the last customer will leave the system, stochastically, at the same time as in $\hat{Q}_n^{tree}$.
\end{lemma}

\begin{proof}
Consider the two following facts regarding the network $\hat{Q}_n^{tree}$.
First, a customer entering the level $l$ will be serviced after all the customers that arrived to the level $l$ before it, are serviced.
Second, at any given moment, only one customer is being serviced at the level $l$ (if there is at least one customer at the nodes $Z_j^l$). These facts are true due to the scheduling of the servers in $\hat{Q}_n^{tree}$ (Definition \ref{dfn:network_tn_hat}).

Clearly, the same facts are true for the network ${Q}_{l_{\max}}^{line}$. First, any customer entering to the level $l$ will be serviced after all the customers that arrived to the level $l$ before it, are serviced. Second, at any given moment, only one customer is being serviced at the level $l$ (if there is at least one customer in the node $Z_1^l$). These facts are true since in ${Q}_{l_{\max}}^{line}$, at every level, there is a single queue with a single server (Definition \ref{dfn:line_of_queues}).

%So, whenever a customer is in service in the network $T_n$, it will be also in service in the network $L_n$.
%So, from the customers' perspective the networks ${Q}_{l_{\max}}^{line}$, and the network $\hat{Q}_n^{tree}$ are the same, i.e.,
So, the departure times of every customer from every level $l$ ($l\in[1,...,l_{\max}]$) are, stochastically, the same in both systems.
The departures from level $l=1$ are the departures from the node $Z_1^1$, and thus the lemma holds.
\end{proof}

Now we are going to move one customer, one queue backward and will show that the resulting system will have stochastically larger (or the same) stopping time.

\begin{lemma}
\label{lemma:move_one_customer_back}
Consider a network ${Q}_{l_{\max}}^{line}$. Let $m$ be a level index: $m\in{[1,..,l_{\max}-1]}$. We take the last customer at the node $Z_1^m$ and place it at the head of the queue of the node $Z_1^{m+1}$, and call the resulting network -- $\grave{Q}_{l_{\max}}^{line}$ (Fig. \ref{fig:moving_customers_in_line} (b)).
%Let $d(Z_1^l)$ and $\grave{d}(Z_1^l)$ be a sequences of departures from the node $Z_1^l$ in ${Q}_{l_{\max}}^{line}$ and in $\grave{Q}_{l_{\max}}^{line}$ respectively.
Then:
\begin{align}
{d}_i(Z_1^1,{Q}_{l_{\max}}^{line})\preceq \grave{d}_i(Z_1^1,\grave{Q}_{l_{\max}}^{line}) \text{ } \forall i\in[1,...,k].
\end{align}
Thus, in $\grave{Q}_{l_{\max}}^{line}$, the last customer will leave the system, stochastically, at the same time or later than in ${Q}_{l_{\max}}^{line}$, or: $t({Q}_{l_{\max}}^{line})\preceq t(\grave{Q}_{l_{\max}}^{line})$.
\end{lemma}

\begin{proof}
We call the customer that was moved -- customer $c$.
Let us take a look at the times of arrivals to the node $Z_1^{m}$ in ${Q}_{l_{\max}}^{line}$ and in $\grave{Q}_{l_{\max}}^{line}$.
Since the customer $c$ is already located in the queue of $Z_1^m$ in ${Q}_{l_{\max}}^{line}$, its arrival time can be considered as $0$.
%\begin{itemize}
%\item The first arrival to $Z_1^{m}$ in ${Q}_{l_{\max}}^{line}$ can be considered as the arrival of $c$ at $a_1(Z_1^{m})=0$ (since $c$ is already in $Z_1^{m}$).
%\item The first arrival to $Z_1^{m}$ in ${Q}_{l_{\max}}^{line}$ will occur at $\grave{a}_1(Z_1^{m})=X_1$, where the r.v. $X_1$ is the service time of the customer $c$.
%\item The $i$'s arrival to $Z_1^{m}$ in ${Q}_{l_{\max}}^{line}$ will occur at $a_i(Z_1^{m})=X_2$, where the r.v. $X_1$ is the service time of the customer $c$.
%\end{itemize}
In $\grave{Q}_{l_{\max}}^{line}$, the arrival time of $c$ is at least $0$ (it should be serviced at $Z_1^{m+1}$ before arriving at $Z_1^m$). Each one of the rest customers that should arrive at $Z_1^m$ will arrive in $\grave{Q}_{l_{\max}}^{line}$, stochastically, at the same time or later than in ${Q}_{l_{\max}}^{line}$, since in $\grave{Q}_{l_{\max}}^{line}$ the server $Z_1^{m+1}$ should first service the customer $c$, and only then will start servicing the rest customers. Thus, ${d}_i(Z_1^{m+1},\grave{Q}_{l_{\max}}^{line})\succeq d_i(Z_1^{m+1},{Q}_{l_{\max}}^{line})$. Using Lemma \ref{lemma:later_arrivals} we obtain that: ${d}_i(Z_1^{m},\grave{Q}_{l_{\max}}^{line})\succeq d_i(Z_1^{m},{Q}_{l_{\max}}^{line})$. Iteratively applying Lemma \ref{lemma:later_arrivals} to the nodes $Z_1^l$, $l\in[m-1,...,1]$, we obtain the result: ${d}_i(Z_1^{1},\grave{Q}_{l_{\max}}^{line})\succeq d_i(Z_1^{1},{Q}_{l_{\max}}^{line})$.
\end{proof}

%\begin{figure*}
%\centering
%\includegraphics[width=3.5in,clip=true, viewport=0.3in 2.5in 8.8in 8.2in]{fig/queue_lemas2.pdf}
%\caption{(a) -- Network ${Q}_{l_{\max}}^{line}$.
%(b) -- Network $\grave{Q}_n^{tree}$, where only one customer is moved one queue backward.
%(c) -- Network $\hat{Q}_{l_{\max}}^{line}$, where all the customers are at the last queue.}
%\label{fig:fig/queue_lemas2.pdf}
%\end{figure*}

\begin{figure*}
\centering
\input{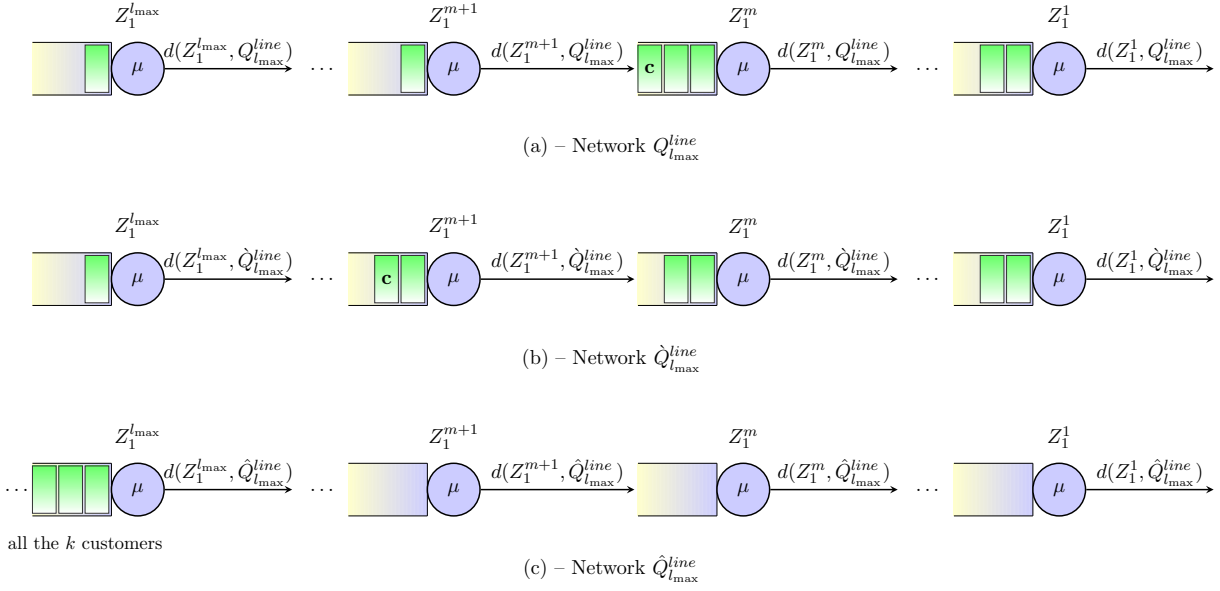}
\caption{(a) -- Network ${Q}_{l_{\max}}^{line}$.
(b) -- Network $\grave{Q}_n^{tree}$, where one customer is moved one queue backward.
(c) -- Network $\hat{Q}_{l_{\max}}^{line}$, where all the customers are at the last queue.}
\label{fig:moving_customers_in_line}
\end{figure*}

\begin{corollary}
\label{corollary:take_all_customers_back}
Consider a network $\hat{Q}_{l_{\max}}^{line}$ (Definition \ref{dfn:line_of_queues_all_customer_back}) which is identical to the network ${Q}_{l_{\max}}^{line}$ with the following change. In $\hat{Q}_{l_{\max}}^{line}$, all the $k$ customers are located at the node $Z_1^{l_{\max}}$ (Fig. \ref{fig:moving_customers_in_line} (c)).
%Let $d_i(Z_1^l)$ and $\hat{d}_i(Z_1^l)$ be a sequences of departures from the node $Z_1^l$ in ${Q}_{l_{\max}}^{line}$ and in $\hat{Q}_{l_{\max}}^{line}$ respectively.
Then:
\begin{align}
{d}_i(Z_1^1,{Q}_{l_{\max}}^{line})\preceq {d}_i(Z_1^1,\hat{Q}_{l_{\max}}^{line}) \text{ } \forall i\in[1,...,k].
\end{align}
Thus, in $\hat{Q}_{l_{\max}}^{line}$, the last customer will leave the system, stochastically, at the same time or later than in ${Q}_{l_{\max}}^{line}$, or: $t({Q}_{l_{\max}}^{line})\preceq t(\hat{Q}_{l_{\max}}^{line})$.
\end{corollary}

\begin{proof}
Given the network ${Q}_{l_{\max}}^{line}$ we take one customer from the tail of some queue (except the queue of the node $Z_1^{l_{\max}}$) and place it at the head of the queue of the preceding node in the ${Q}_{l_{\max}}^{line}$. According to the Lemma \ref{lemma:move_one_customer_back}, we get a network in which every customer leaves via $Z_1^1$, stochastically, not earlier than in ${Q}_{l_{\max}}^{line}$. Iteratively moving customers (one customer and one queue at a time) backwards we get finally the network $\hat{Q}_{l_{\max}}^{line}$ in which all the $k$ customers are located at the node $Z_1^{l_{\max}}$. Since at each step, according to Lemma \ref{lemma:move_one_customer_back}, the departure times from $Z_1^1$ could only get, stochastically, larger, the lemma holds.
\end{proof}

\begin{corollary}
\label{corollary:tree_is_slower_than_line}
The time it will take the last customer to leave the network of $n$ queues arranged in a \emph{tree} topology is, stochastically, the same or smaller than in the network of $n$ queues arranged in a \emph{line} topology where all the $k$ customers are located at the farthest queue, i.e., $t({Q}_n^{tree})\preceq t(\hat{Q}_{l_{\max}}^{line})$.
\end{corollary}

\begin{proof}
This corollary is a direct consequence of the Lemmas \ref{lemma:idle_server}, \ref{lemma:tree_as_line}, and the Corollary \ref{corollary:take_all_customers_back}.
\end{proof}

Now we are ready for the last step of the proof. We will find the stopping time of a system of queues arranged in a line topology and with all the customers located at the last queue.

\begin{lemma}
\label{lemma:line_stopping_time_with_k}
The time it will take to the last customer to leave the system $\hat{Q}_{l_{\max}}^{line}$ ($l_{\max}$ MM1 queues arranged in a line topology) is $O((k+\log n+l_{\max})/\mu)$ with probability of at least $1-\tfrac{1}{n^2}$.
\end{lemma}

\begin{proof}
Initially, all the customers (from now we will call them \emph{real} customers) are located in the last ($Z_1^{l_{\max}}$) queue. We now take all the \emph{real} customers out of this queue and will make them enter the system (via the $Z_1^{l_{\max}}$) from outside. We define the \emph{real} customers' arrivals as a Poisson process with rate $\lambda= \frac{\mu}{2}$. So, $\rho=\frac{\lambda}{\mu}=\frac{1}{2}<1$ for all the queues in the system. Clearly, such an assumption only increases the stopping time of the system (stopping time is the time until the last customer leaves the system).
According to Jackson's theorem, which proof can be found in \cite{Chen2001Fundamentals}, there exists an equilibrium state. So, we need to ensure that the lengths of all queues at time $t=0$ are according to the equilibrium state probability distribution. We add \emph{dummy} customers to all the queues according to the stationary distribution. By adding additional \emph{dummy} customers
%(we call them \emph{dummy} since their arrivals are not counted as a rank increment)
to the system, we make the \emph{real} customers wait longer in the queues, thus increasing the stopping time.

We will compute the stopping time $t(\hat{Q}_{l_{\max}}^{line})$ in two phases:
%First, we will find the time it takes the $k$'th (last) \emph{real} customer to arrive at the rightmost node, i.e., node $Z_1^{l_{\max}}$. By that time, the rank of node $Z_1^1$ will become $k$ and it will finish the algebraic gossip protocol (i.e., it received $k$ \emph{helpful messages}).
Let us denote this time as $t_1+t_2$, where $t_1$ is the time needed for the $k$'th customer to arrive at the first queue, and $t_2$ is the time needed for the $k$'th customer to pass through all the $l_{\max}$ queues in the system.

From Jackson's Theorem, it follows that the number of customers in each queue is independent, which implies that the random variables that represent the waiting times in each queue are independent.
To continue with the proof we need the following lemmas; the first is a classical result from queuing theory, the proof of the second lemma is omitted.
\begin{lemma}[\cite{1378238}, section 4.3]
\label{lemma:waiting_distr}
Time needed to cross one $M/M/1$ queue in the equilibrium state has an exponential distribution with parameter
$\mu -\lambda$.
\end{lemma}
\begin{lemma}\label{lemma:sum_of_exp_bounded1}
Let $Y$ be the sum of $n$ independent and identically distributed exponential random variables.
Then, for $\alpha>1$:
\begin{align}
\Pr \left(Y < \alpha\text{E}\left[Y\right]\right) > 1-(2e^{-\alpha/2})^n.
\end{align}
\end{lemma}

The random variable $t_1$ is the sum of $k$ independent random variables distributed exponentially with parameter $\mu/2$. From Lemma \ref{lemma:waiting_distr} we obtain that $t_2$ is the sum of $l_{\max}$ independent random variables distributed exponentially with parameter $\mu-\lambda=\mu/2$.
$\text{E}\left[t_1\right]=\sum_{i=1}^k 2/\mu=2k/\mu$, and by taking $\alpha=2+4\tfrac{\ln n}{k}$, we obtain:
\begin{align}
\\ \Pr \left(t_1 < (4k+8\ln n)/\mu\right) &> 1-(2e^{-(2+4\tfrac{\ln n}{k})/2})^k
\\&=1-(\tfrac{2}{e})^k e^{-2\ln n}
\\&\ge 1-e^{-2\ln n}
\\&\ge 1-\tfrac{1}{n^2}.
\end{align}
In a similar way we obtain:
\begin{align}
\Pr \left(t_2 < (4l_{max}+8\ln n)/\mu\right)&>1-\tfrac{1}{n^2}.
\end{align}
$t(\hat{Q}_{l_{\max}}^{line})=t_1+t_2$, thus, using union bound:
\begin{align}
&\Pr \left(t_1+t_2 < (4k+4l_{max}+16\ln n)/\mu\right)>1-\tfrac{2}{n^2}
\\\notag &\text{and thus:}
\\ & t(\hat{Q}_{l_{\max}}^{line})=O((k+l_{\max}+\log n)/\mu) \\\notag&\text{ w.p. of at least }1-\tfrac{2}{n^2}.
\end{align}
\end{proof}

From Claim \ref{corollary:tree_is_slower_than_line} we obtain that $t({Q}_n^{tree})\preceq t(\hat{Q}_{l_{\max}}^{line})$ and thus: $t({Q}_n^{tree})= O((k+l_{\max}+\log n)/\mu)$ w.p. of at least $1-\tfrac{2}{n^2}$.
\end{proof}

\subsection*{Proof of Theorem \ref{thm:linear_broadcast}}
\addcontentsline{toc}{subsection}{Proof of Theorem \ref{thm:linear_broadcast}}

\begin{thm:linear_broadcast}[restated]
\label{thm:linear_broadcast_restated}
For any connected graph $G_n$, the stopping time of the broadcast protocol with the round-robin communication model -- $\mathcal{B_{RR}}$ is $O(n)$ rounds. In the asynchronous time model, this result holds with probability of at least $1-n(2/e)^{3n}$, and in the synchronous time model, with probability $1$.
\end{thm:linear_broadcast}

For the proof we need the following lemma:

\begin{lemma}
\label{lemma:geom_upper_bound}
Let $X$ be a sum of $m$ independent and identically distributed geometric random variables (each one with parameter $p>0$) and $\text{E}\left[X\right]=\tfrac{m}{p}$.
Then, for $\alpha>1$:
\begin{align}
\Pr \left(X \leq \alpha\text{E}\left[X\right]\right) > 1-\left(\alpha e^{1-\alpha}\right)^m.
\end{align}
\end{lemma}

\begin{proof}
First, we will define $Y$ as the sum of $k$ independent Bernoulli random variables, i.e., $Y=\sum_{i=1}^{k}Y_i$, where $Y_i \sim Bernoulli(p)$.
Let us notice that:
\begin{align}
\Pr \left(X \leq k\right)= \Pr \left(Y \geq m\right)
\end{align}
The last is true since the event of observing at least $m$ successes in a sequence of $k$ Bernoulli trials implies that the sum of $m$ independent geometric random variables is no more than $k$. On the other hand, if the sum of $m$ independent geometric random variables is no more than $k$ it implies that $m$ successes occurred no later than the $k$-th trial and thus $Y\geq m$.

Now we will use a Chernoff bound for the sum of independent Bernoulli random variables presented in \cite{1076315}:
For any $0<\delta < 1$ and $\mu = \text{E}\left[Y\right]$:
\begin{align}
\Pr \left(Y \le (1-\delta)\mu\right) \le \left(\frac{e^{-\delta}}{(1-\delta)^{1-\delta}}\right)^\mu.
\end{align}
Since $\mu=\text{E}\left[Y\right]=kp$, and by letting $\delta=\frac{kp-m}{kp}$ we obtain:
\begin{align}
\Pr \left(Y \le (1-\delta)\mu\right) = \Pr \left(Y \le m \right) &\le \left(\frac{m}{e^{\frac{m-kp}{m}}kp}\right)^{-m}.
\\\Pr \left(Y \ge m\right) &> 1-\left(\frac{m}{e^{\frac{m-kp}{m}}kp}\right)^{-m}
\\\Pr \left(X \leq k\right) &> 1- \left(\frac{m}{e^{\frac{m-kp}{m}}kp}\right)^{-m}
\end{align}

By substituting $k=\alpha\tfrac{m}{p}=\alpha\text{E}\left[X\right]$ (where $\alpha>1$) we obtain:
\begin{align}
\Pr \left(X \leq \alpha\text{E}\left[X\right]\right) > 1-\left(\frac{e^{\alpha}}{e\alpha}\right)^{-m}
\end{align}
\end{proof}

\begin{proof}[Proof of Theorem \ref{thm:linear_broadcast}]
In this proof we assume the \push gossip variation, but it is clear that the result holds also for \ex.

Without loss of generality, assume that the message that needs to be disseminated is initially located at the node $v$.
In the \emph{round-robin} gossip, when a node is scheduled to transmit, it transmits a message to its neighbor according to the \emph{round robin} scheme. I.e, at every transmission a message is sent to a different neighbor.

Consider a shortest path between $v$ and some other node $u$.
On the shortest path of length $l$ there is exactly one node at the distance $i$ from $v$, where $i\in\left[0,\dots,l\right]$, and $l\le n-1$.
Let $d_i$ be the degree of a node at the distance $i$ from $v$.
In order to guarantee the delivery of the message from $v$ to $u$, we need $\sum_{i=0}^{l}d_i$ transmissions in the following order: first, we need at $d_0$ transmissions of the node $v$, then $d_1$ transmissions of the next node in the path $v\rightarrow u$, and so on until the message is delivered to $u$. From Theorem \ref{lemma:sum_of_degrees_on_shortest_path}, $\sum_{i=0}^{l}d_i\le 3n$.

In the asynchronous model, a node transmits at a given \emph{timeslot} with probability $\tfrac{1}{n}$.
So, the number of timeslots until some specific node transmits is a geometric random variable with parameter $\tfrac{1}{n}$. We define this geometric random variable as $X$, i.e., $X\sim\text{Geom} \left(\tfrac{1}{n}\right)$.

The number of timeslots until $3n$ specific transmissions occur, is the sum of $3n$ independent geometric random variables. Using Lemma \ref{lemma:geom_upper_bound} we obtain the bound of $O(n^2)$ timeslots (or $O(n)$ rounds) with exponential high probability. The last allows us to perform union bound for shortest paths to all other nodes in $G$, thus obtaining the $O(n)$ bound for the broadcast time.

Now, let us write the above more formally:

We define $t_{v\rightarrow u}$ as the time it takes to guarantee the delivery of the message from the node $v$ to an arbitrary node $u$. As we showed above, $t_{v\rightarrow u}$ is the number of timeslots until $3n$ specific transmissions occur, so:
\begin{align}
t_{v\rightarrow u}=\sum_{i=1}^{3n}X_i ,
\end{align}
\begin{align*}
X_i \text{ } \forall i\in[1,\ldots,3n] \text{ are i.i.d. and distributed as }X,\\\text{where }X\sim\text{Geom} \left(\frac{1}{n}\right).
\end{align*}
Thus,
\begin{align}
\text{E}\left[t_{v\rightarrow u}\right]=\text{E}\left[\sum_{i=1}^{3n}X_i\right]=\sum_{i=1}^{3n}\text{E}\left[X_i\right]=\sum_{i=1}^{3n}n = 3n^2.
\end{align}
From Lemma \ref{lemma:geom_upper_bound} with $\alpha = 2$ :
\begin{align}
\Pr\left(t_{v\rightarrow u}\leq 2\text{E}\left[t_{v\rightarrow u}\right]\right)>1-(2/e)^{3n},
\end{align}
or
\begin{align}
\Pr\left(t_{v\rightarrow u}\leq 6n^2\right)>1-(2/e)^{3n}.
\end{align}
Now, we will apply a union bound on probabilities of the events: $t_{v\rightarrow u'}>6n^2$, where $u'\in{V}$. Notice, that $\text{E}\left[t_{v\rightarrow u}\right]=3n^2$ for all $u'\in{V}$.
\begin{align}
\Pr\left(\bigcup_{u'\in{V}}(t_{v\rightarrow u'}>6n^2)\right)\le \sum_{u'\in{V}}\Pr \left(t_{v\rightarrow u'}\right),
\end{align}
so,
\begin{align}
\Pr\left(\bigcup_{u'\in{V}}(t_{v\rightarrow u'}>6n^2)\right)\le n(2/e)^{3n}.
\end{align}
Thus,
\begin{align}
\Pr\left(\bigcap_{u'\in{V}}(t_{v\rightarrow u'}\le 6n^2)\right)>1-n(2/e)^{3n}.
\end{align}
So, we obtain the result of $O(n^2)$ timeslots, or $O(n)$ rounds.

Easy to see that in the synchronous time model, $3n$ specific transmissions will occur exactly after $3n$ communication rounds. E.g., after $d_0$ rounds, $v$ will perform $d_0$ transmissions -- each one to different neighbor (according to the round-robin scheme). Thus, the message will be delivered to $u$ after at most $3n$ rounds with probability $1$.
\end{proof}

\subsection*{Proof of Theorem~\ref{theorem:async-info-spr}}
\addcontentsline{toc}{subsection}{Proof of Theorem \ref{theorem:async-info-spr}}
\begin{theorem:async-info-spr}[restated]
\label{theorem:async-info-spr-restated}
Let $c=O(\log^p{(n)})$ for some $p\geq 0$, let $G$ be a graph with weak conductance $\Phi_c=\Omega(\frac{1}{\log^p{(n)}})$, and let $k=\Omega(\log^{2p+3}{(n)})$. With probability at least $1-\frac{1}{n}$, the time for disseminating $k$ messages using protocol TAG in conjunction with the IS protocol is $O(k+l_{\max})$ rounds for the asynchronous time model, where $l_{\max}$ is the depth of the spanning tree induced by the IS protocol.
\end{theorem:async-info-spr}
\begin{proof}
To simulate one round of a synchronous protocol, we consider the execution of the protocol for $O(n\log{(n)})$ time slots, which is $\log{(n)}$ asynchronous rounds. With high probability, $1-\delta$ for some small $\delta$, each node takes at least one step. This follows from a standard coupon collector's argument, as steps of each node correspond to a different coupon. The crucial property of the information spreading protocol that allows our analysis to go through is its \emph{monotone} nature, that is, the information collected and sent by a node is an $n$-bit string whose entries can only turn form zero to one as time passes. This implies that whenever each node took at least one step, the strings obtained can only contain more one entries than the strings obtained by one round of the synchronous model (recall that the goal is for all nodes to obtain a string of ones). Hence, after $O(T\log{(n)})$ asynchronous rounds, the information the nodes have is at least the information that they have after $T$ rounds in the synchronous model. This does not yet conclude the proof, for the following reason. The analysis of the synchronous protocol goes through in this simulation except for one argument~\protect{\cite[Claim 1]{censor2010fast}}, which bounds the size of the deterministic list of subset of neighbors that is maintained by a node $v$. This size is bounded by the number of steps taken by $v$. On one hand, we need the number of steps taken by $v$ in each $O(n\log{(n)})$ time slots to be at least one to argue the simulation, but on the other hand, it may be that a node takes a larger number of steps. This would imply that its list is larger than in the corresponding synchronous case. However, returning to the coupon collector's problem, we have that actually no node takes more than $O(\log{(n)})$ steps in each $O(n\log{(n)})$ time slots (within the same high probability). This implies that the size of the list a node maintains is at most a multiplicative factor of $O(\log{(n)})$ larger than its size after $T$ synchronous rounds. Since this size is used, in turn, to bound the number of synchronous rounds required, we have to add an additional $O(\log{(n)})$ factor to the number of rounds in the asynchronous model.

Finally, we note that the probability of failure of the coupon collector's argument (in either the lower or upper bound on the number of steps per node) needs to be added up for all simulated rounds. When this number of rounds $T$ is polylogarithmic in $n$, we have that using a union bound we remain with a high probability for the entire argument.
\end{proof}

\end{document}